\documentclass[preprint,12pt]{elsarticle}
\usepackage[utf8]{inputenc}
\usepackage{amsthm}
\usepackage{amsmath}
\usepackage{amsfonts}
\usepackage{amssymb}
\usepackage{graphicx}
\usepackage{epstopdf}
\usepackage{color}
\usepackage{multirow}
\usepackage{mathtools}
\usepackage{xcolor}
\usepackage{a4wide}
\usepackage{bm}
\usepackage{caption}
\usepackage{subcaption}
\usepackage[
  separate-uncertainty = true,
  multi-part-units = repeat
]{siunitx}
\usepackage{rotating}
\usepackage{verbatim}
\usepackage{lineno,hyperref}
\modulolinenumbers[5]
\usepackage{tikz}
\tikzstyle{noeud} = [circle, draw, fill=white, inner sep=2pt]
\journal{Theoretical Computer Science}
\newtheorem{theorem}{Theorem}[section]
\newtheorem{lemma}{Lemma}[section]

\newdefinition{remark}{Remark}[section]

\newdefinition{definition}{Definition} \newdefinition{example}{Example}

\usepackage[noend,linesnumbered,ruled,lined]{algorithm2e}
\SetAlFnt{\scriptsize}
\usepackage{amsfonts}
\usepackage{makecell}

\usepackage[usenames,dvipsnames]{pstricks}
\usepackage{pstricks-add}
\usepackage{epsfig}
\usepackage{pst-grad} 
\usepackage{pst-plot} 
\usepackage[space]{grffile} 
\usepackage{etoolbox} 
\makeatletter 
\patchcmd\Gread@eps{\@inputcheck#1 }{\@inputcheck"#1"\relax}{}{}
\makeatother
\usepackage{mathtools}

\usepackage{graphicx}
\usepackage{tabularx}
\usepackage{textcomp}
\usepackage{xcolor}
\usepackage[labelformat=simple]{subcaption}

\usepackage{booktabs}

\usepackage{pgf,tikz,pgfplots}
\pgfplotsset{compat=1.15}
\usepackage{mathrsfs}
\usepackage{hyperref}
\usepackage{multicol}

\newtheorem{myclaim}{Claim}
\newtheorem{problem}{Problem}

\SetCommentSty{mycommfont}

\journal{ArXiv}

\usepackage{a4wide}
\begin{document}
\usetikzlibrary{arrows}
\begin{frontmatter}

\title{When Agents are Powerful: \\Black Hole Search with Verification in Time-Varying Graphs \footnote{A preliminary version \cite{Ashish_2026} was accepted at ICDCIT~2026.}}

\author[1]{Tanvir Kaur}
\author[1]{Ashish Saxena \footnote{Corresponding author}}
\affiliation[1]{organization={Department of Mathematics},
            addressline={Indian Institute of Technology Ropar}, 
            city={Rupnagar},
            postcode={140001}, 
            state={Punjab},
            country={India}}

\begin{abstract}
A black hole is a harmful node in a graph that destroys any agent entering it, making its identification a critical task. In the \emph{Black Hole Search with Verification (BHSV)} problem, a team of agents operates on a graph $G$ with the objective that at least one agent survives and correctly identifies an edge incident to the black hole; if no black hole exists, then all agents must terminate. 

Prior work has studied BHS in arbitrary dynamic graphs under the restrictive \emph{face-to-face} communication model, where agents can exchange information only when co-located. This constraint significantly increases the number of agents required to solve the problem. In this work, we strengthen the capabilities of agents by equipping them with (i) \emph{1-hop visibility}, (ii) \emph{global communication}, and (iii) both \emph{1-hop visibility} and \emph{global communication}. We show that these enhancements lead to more efficient solutions for the BHSV problem in dynamic graphs.
\end{abstract}


\end{frontmatter}
\section{Introduction}
In many distributed systems, the network cannot be assumed to be fully reliable. Real-world systems are prone to faults: agents may crash, communication links may intermittently fail, or nodes may behave maliciously, including corrupting data or destroying visiting agents. One particularly dangerous fault is a \emph{hostile node} that destroys any agent entering it without leaving any trace of its destruction. Such a node, denoted by $v_{BH}$, is called a \emph{black hole}. In the literature, two variants of the \emph{Black Hole Search (BHS)} problem have been studied: 
(i) at least one agent must survive and produce a map of the network indicating all edges leading to $v_{BH}$, and 
(ii) at least one agent must survive and learn \emph{at least one} edge that leads to $v_{BH}$. 
The BHS problem has been extensively studied on static graphs \cite{bhs_tokens, complexity_black_hole, Paola_2006}.

Recently, researchers have begun to investigate fundamental distributed problems such as exploration \cite{Nicolo_21}, dispersion \cite{dynamic_dispersion}, and gathering \cite{LunaTCS2020} in \emph{dynamic graphs}, which more accurately model networks whose connectivity changes over time. In the synchronous setting, time is divided into discrete rounds, and a dynamic graph $\mathcal{G}$ is represented as a sequence of static graphs $\mathcal{G}_0, \mathcal{G}_1, \mathcal{G}_2, \ldots$, where $\mathcal{G}_r$ denotes the network at round $r$. This evolving representation is referred to as a \emph{time-evolving} or \emph{dynamic} graph. A natural question is whether the first variant of the BHS problem remains solvable in dynamic graphs. The answer is negative: since an edge leading to $v_{BH}$ may not appear in every round $r \ge 0$, it is impossible to guarantee identification of all such edges. Therefore, only the second variant of the BHS problem remains meaningful in dynamic settings. This variant has been considered in previous work on dynamic graphs \cite{Luna_2025, bhattacharya_2023,Adri_tori, BHS_gen}. While several results are known for restricted classes of dynamic graphs, only two works address the case of \emph{general} dynamic graphs \cite{BHS_gen,Tanvir_Scat}. In that work, it is ensured that at least one agent reaches a neighbor of $v_{BH}$ and identifies a port leading to it.

The existing literature on BHS implicitly assumes that the network contains exactly one black hole node. However, in practice, it may be the case that no such destructive node is present at all. This raises an important question: how should the problem be formulated when the presence of a black hole itself is not guaranteed? Addressing this question leads to a broader formulation of the search task, one that must allow agents not only to detect a black hole when it exists. Otherwise, all agents must terminate. Based on that, we have the following definition.
\begin{definition}[Black Hole Search with Verification (BHSV)]
\emph{A black hole $v_{BH}$ is said to be found if there exists a black hole in the network and at least one agent present at an adjacent node (to $v_{BH}$), say $v$ identifies a port $p$ associated with $v$ leading to $v_{BH}$. If no black hole exists in the network, then all agents must terminate.}
\end{definition}

In \cite{Luna_2025}, it is shown that two agents cannot determine the location of the black hole $v_{BH}$ in 1-interval connected rings, whereas three agents are sufficient to locate $v_{BH}$ in 1-interval connected dynamic rings. This naturally leads to the following question: using three agents, is it possible to guarantee that at least one agent reaches a neighbour of $v_{BH}$ and correctly identifies the port leading to $v_{BH}$? The \underline{motivation} behind this is the following. If one of the agents can correctly identify the port that leads to $v_{BH}$, then that agent can be endowed with the ability to repair the network. On the other hand, if no black hole is present in the network, the agents should be able to correctly recognize this and terminate. The \underline{motivation} here is that, in existing algorithms, agents continue moving indefinitely while searching for $v_{BH}$. In the case where no black hole exists, such perpetual movement wastes time and energy without contributing to the task. Our problem definition avoids this unnecessary activity by allowing agents to terminate when the network is safe.

Most previous work on the BHS problem assumes the \textit{face-to-face} (f-2-f) communication model, where agents can communicate only when co-located~\cite{Luna_2025,bhattacharya_2023,Adri_tori,BHS_gen,Tanvir_Scat}. This restriction limits coordination and often increases the number of agents or the time required to solve the problem. In~\cite{BHS_gen,Tanvir_Scat}, the authors study BHS in dynamic general graphs; however, there remains a significant gap between the known lower and upper bounds on the number of agents. To reduce this gap, we study the BHS problem under two stronger models: (i) the \textit{global communication model}, where agents can communicate regardless of their positions, and (ii) the \textit{1-hop visibility model}, where each agent can observe its immediate neighbourhood, including the presence of agents at adjacent nodes. While these models provide more power than f-2-f communication, they also introduce new challenges. For example, under global communication, agents can exchange information but may not know how to physically reach one another due to the lack of structural knowledge of the network. To the best of our knowledge, this is the first work to study BHSV under communication models beyond f-2-f. In the next section, we present the detailed model and problem definition.

\subsection{Model and the problem}\label{sec:model}
\noindent \textbf{Dynamic graph model}: A dynamic network is modeled as a \emph{time-varying graph (TVG)}, denoted by \( \mathcal{G} = (V, E, T, \rho) \), where \( V \) is a set of nodes, \( E \) is a set of edges, \( T \) is the discrete temporal domain, and \( \rho : E \times T \rightarrow \{0, 1\} \) is the presence function, which indicates whether a given edge is available at a given time. The static graph \( G = (V, E) \) is referred to as the underlying graph (or footprint) of the TVG \( \mathcal{G} \), where \( |V| = n \) and \( |E| = m \). For a node \( v \in V \), let \( E(v) \subseteq E \) denote the set of edges incident on \( v \) in the footprint. The degree of node \( v \) is defined as \( \deg(v) = |E(v)| \). Nodes in $V$ are anonymous. Each node is equipped with storage, and each edge incident to a node \( v \) is locally labeled with a port number. This labeling is defined by a bijective function \( \lambda_v : E(v) \rightarrow \{0, \ldots, \delta_v - 1\} \), which assigns a distinct label to each incident edge of \( v \). Since the time is discrete, the TVG \( \mathcal{G} \) can be viewed as a sequence of static graphs \( \mathcal{S}_{\mathcal{G}} = \mathcal{G}_0, \mathcal{G}_1, \ldots, \mathcal{G}_r, \ldots \), where each \( \mathcal{G}_r = (V, E_r) \) denotes the snapshot of \( \mathcal{G} \) at round \( r \), with \( E_r = \{ e \in E \mid \rho(e, r) = 1 \} \). The set of edges not present at time \( r \) is denoted by \( \overline{E}_r = E \setminus E_r \subseteq E \). There may be a node $v_{BH}$ in $G$ which is a black hole, and its degree is denoted by $\delta_{BH}$. A node that is not a black hole, we call it a \underline{safe node}.
Dynamic graphs can be classified based on how their topological changes affect connectivity. A commonly used restriction is 1-interval connectivity, and a further refinement is its bounded variant.

\begin{definition}
     (1-interval connectivity) A dynamic graph \( \mathcal{G} \) is \emph{1-interval connected} (or \emph{always connected}) if every snapshot \( \mathcal{G}_r \in \mathcal{S}_\mathcal{G} \) is connected.
\end{definition}

\begin{definition}
 (\( \ell \)-bounded 1-interval connectivity) A dynamic graph \( \mathcal{G} \) is said to be \( \ell \)-bounded 1-interval connected if it is always connected and \( |\overline{E}_r| \leq \ell\).
\end{definition}

We call a node a \underline{hole} if it has no agent, and call a node a \underline{multinode} if at least two agents are present at that node.

\medskip
\noindent\textbf{Agent}: We consider $k$ agents to be present arbitrarily at safe nodes of the graph $G$ in the initial configuration. Each agent has a unique identifier assigned from the range $[1,\,n^c]$, where $c$ is a constant. Each agent knows its ID but is unaware of the other agents' IDs. Agents are not aware of the values of \( n \), \( k \), or \( c \) unless stated otherwise. The agents are equipped with memory. An agent residing at a node, say $v$, in round $r$ knows $\deg(v)$ in the footprint $G$ and all associated ports corresponding to node $v$ in $G$. We call the initial configuration of agents \underline{rooted} if all $k$ agents are co-located at a node; otherwise, we call it \underline{scattered}.

\medskip
\noindent \textbf{Communication model:} We consider two communication models: (i) face-to-face (f-2-f) communication~\cite{Augustine_2018}, meaning agents can only communicate if they are co-located at the same node, and (ii) global communication~\cite{Ajay_dynamicdisp}, allowing agents to exchange messages regardless of their positions in the network.

\medskip
\noindent \textbf{Visibility model:} We use three types of visibility models: 0-hop visibility, 1-hop visibility and full visibility. In the 0-hop visibility~\cite{Augustine_2018}, an agent at a node \( v \in \mathcal{G} \) can see the IDs of agents present at \( v \) in round \( r \), as well as the port numbers at \( v \), but nothing beyond that. In the 1-hop visibility model~\cite{Agarwalla_2018}, an agent \( a_i \) at node \( v \) can also see all neighbors of \( v \), including the IDs of agents (if any) at those neighboring nodes. Let \( e_v \) be an edge incident to node \( v \). In the 0-hop visibility model, agents cannot determine the value of \( \rho(e_v, r) \) at the beginning of round \( r \). In contrast, under the 1-hop visibility model, they can determine this value at the beginning of round \( r \). In full visibility, at round $r$, the agent can see $\mathcal{G}_r$ as well as agents' positions in $\mathcal{G}_r$. 

\medskip
\noindent The algorithm runs in synchronous rounds. In each round $r$, an agent $a_i$ performs one \textit{Communicate-Compute-Move} (CCM) cycle as follows:
\begin{itemize}
    \item \textbf{Communicate:} Agent $a_i$ at node $v_i$ can communicate with other agents $a_j$ present at the same node $v_i$ or present at any different node $v_j$, depending on the communication model used. The agent also understands whether it had a successful or an unsuccessful move in the last round.
    \item \textbf{Compute:} Based on the information the agent has, the agent computes the port through which it will move or decides not to move at all.
    \item \textbf{Move:} Agent moves via the computed port or stays at its current node. 
\end{itemize}

In this work, we study the following two problems. 

\begin{problem}[1-BHSV]
Let $\mathcal{G}$ be a 1-bounded 1-interval connected dynamic graph. Suppose $k$ agents are initially placed at safe nodes of $\mathcal{G}$. The \emph{1-BHSV} problem is said to be solved if the following holds. If $v_{BH}$ exists in $\mathcal{G}$, then at least one agent is guaranteed to reach a neighbor $v$ of $v_{BH}$ at some round $t$, and identify a port of $v$ at round $t$ that leads to $v_{BH}$. Otherwise, if no $v_{BH}$ exists in $\mathcal{G}$, then all agents must terminate.
\end{problem}

\begin{problem}[BHSV in 1-interval connected]
Let $\mathcal{G}$ be a 1-interval connected dynamic graph. Suppose $k$ agents are initially placed at safe nodes of $\mathcal{G}$. The \emph{BHSV} problem in 1-interval connected is said to be solved if the following holds. If $v_{BH}$ exists in $\mathcal{G}$, then at least one agent is guaranteed to reach a neighbor $v$ of $v_{BH}$ at some round $t$, and identify a port of $v$ at round $t$ that leads to $v_{BH}$. Otherwise, if no $v_{BH}$ exists in $\mathcal{G}$, then all agents must terminate.
\end{problem}

In the next section, we present the existing work on BHS.

\subsection{Related work}
The BHS problem was first introduced by Dobrev et al.~\cite{Paola_2006} and has been extensively studied in static graphs~\cite{Pelc_2005, Paola_2010, Shantanu_2011, MarkouP12}. 

More recently, attention has shifted to dynamic graphs, where the network topology evolves over time. In this setting, most work focuses on $1$-BHS and $f$-BHS in $f$-bounded $1$-interval connected graphs with a single black hole. These problems have been studied for specific graph classes such as rings~\cite{Luna_2025}, cactus graphs~\cite{bhattacharya_2023}, and tori~\cite{Adri_tori}, and more recently for general graphs~\cite{BHS_gen}. The authors of~\cite{BHS_gen} show that $1$-BHS requires at least $2\delta_{BH}$ agents in the scattered setting, and that no solution exists for $f$-BHS with $2f+1$ co-located agents, even with unbounded memory. However, their algorithms assume a rooted initial configuration. In~\cite{Tanvir_Scat}, the authors study a more general scattered setting for $1$-BHS using $2\delta_{BH}+17$ agents. In this work, we further extend the study of $1$-BHS on arbitrary graphs by equipping agents with stronger capabilities and obtain optimal bounds on the number of agents required.

\subsection{Our contribution}
In this work, we establish the following six results:
\begin{enumerate}
    \item It is impossible for $3$ agents that are co-located at a safe node of the graph $G$ with $n$ nodes to solve the problem of 1-BHSV even if the agents have 1-hop visibility, global communication and infinite memory, and the nodes have infinite storage (refer to Theorem~\ref{thm:imp1}).
    \item It is impossible for $\delta_{BH}+1$ agents, initially placed at safe nodes of the graph, to solve the \textnormal{1-BHSV} problem, even if the agents have 1-hop visibility, global communication, infinite memory, and the nodes have infinite storage (refer to Theorem~\ref{thm:imp2}).

    \item It is impossible for $n-2$ agents, initially located at arbitrary safe nodes, to solve the BHSV problem in $1$-interval connected graphs, even if the agents have full visibility, global communication, and infinite memory, and the nodes have infinite storage (refer to Theorem \ref{thm:imp3}).
    \item We design an algorithm that solves 1-BHSV using four agents starting from a rooted configuration, where each agent has 1-hop visibility, f-2-f communication, and $O(\log n)$ memory, and each node has $O(\log n)$ storage (Theorem~\ref{thm:1-hop}).
    \item We design an algorithm that solves 1-BHSV using $\delta_{BH}+2$ agents starting from any configuration, where each agent has global communication, 0-hop visibility, and $O(\log n)$ memory, and each node has $O(\log n)$ storage (Theorem~\ref{thm:global}).
    \item We design an algorithm that solves BHSV in 1-interval connected graphs using $n-1$ agents starting from any configuration, where each agent has global communication, 1-hop visibility, and $O(\log n)$ memory, and each node has no storage (Theorem~\ref{thm:global-1-hop}).
\end{enumerate}

\section{Impossibility results}
In this section, we present the impossibility results based on the initial configuration of agents: (i) co-located agents and (ii) scattered agents.

\begin{figure}[h]
    \centering
    \includegraphics[width=0.5\linewidth]{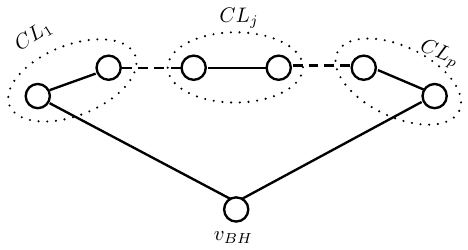}
    \caption{The construction of $G$.}
    \label{fig:imp1}
\end{figure}

\begin{figure}
    \centering
    \includegraphics[width=1\linewidth]{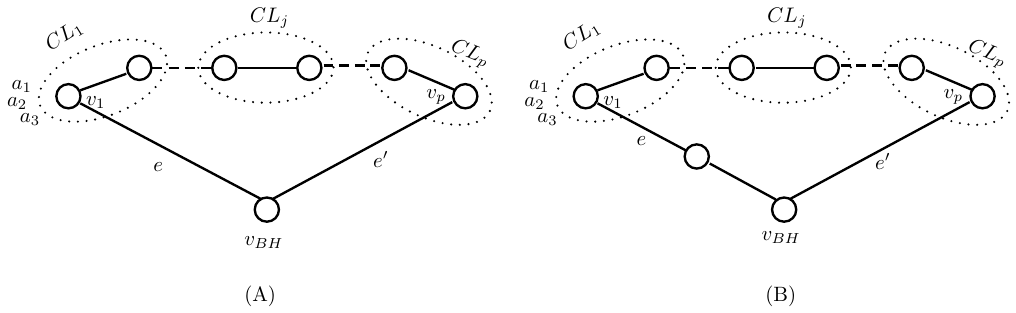}
    \caption{(A) Footprint graph $G$. (B) Modified footprint graph $G'$ if the agents declare, without traversing edge $e$, that edge $e$ in $G$ leads to $v_{BH}$.}
    \label{fig:imp2}
\end{figure}

In the following result, we show that three co-located agents cannot solve 1-BHSV. The idea is to construct a graph together with an adversarial strategy that prevents the agents from safely identifying the black hole without sacrificing agents. The graph consists of a sequence of small cliques connected in a path-like manner, with the black hole attached at both ends. The adversary dynamically removes edges so that whenever multiple agents are near a potential entry to the black hole, access to it is blocked. Consequently, the agents cannot verify whether an edge leads to the black hole without sending at least one agent to traverse it. However, once an agent attempts such a traversal and is destroyed, the adversary ensures that the remaining agents are separated and confined away from both possible entry points, preventing them from completing the verification. Hence, more than three agents are required to solve 1-BHSV under this model.

\begin{theorem}\label{thm:imp1}
($n\geq 9$) It is impossible for $3$ agents that are co-located at a safe node of the graph $G$ with $n$ nodes to solve the problem of 1-BHS even if the agents have 1-hop visibility, global communication and infinite memory, and the nodes have infinite storage.
\end{theorem}
\begin{proof}
Let the size of the graph \( G \) be \( n \), and without loss of generality assume \( n - 1 = 2p \) for some integer \( p \geq 3 \). Construct \( G \) as follows. There are \( p \) cliques \( CL_1, CL_2, \ldots, CL_p \), each of size \( 2 \), and an additional node \( v_{BH} \) representing the black hole. Let \( v_1 \in CL_1 \) and \( v_p \in CL_p \), and add edges \( e = (v_{BH}, v_1) \) and \( e' = (v_{BH}, v_p) \). To interconnect the cliques, define connector nodes as follows. Let \( w_1^{(1)} \in CL_1 \) with \( w_1^{(1)} \neq v_1 \); for \( 2 \leq i \leq p-1 \), let \( w_1^{(i)}, w_2^{(i)} \in CL_i \); and let \( w_1^{(p)} \in CL_p \) with \( w_1^{(p)} \neq v_p \). Define edges \( e_1 = (w_1^{(1)}, w_1^{(2)}) \), \( e_i = (w_2^{(i)}, w_1^{(i+1)}) \) for \( 2 \leq i \leq p - 2 \), and \( e_{p-1} = (w_2^{(p-1)}, w_1^{(p)}) \). Refer to Figure~\ref{fig:imp1} for \( n = 7 \).

Assume that \( G \) is the footprint graph and that all three agents are initially co-located at node \( v_1 \). The adversary can remove at most one edge per round and follows the following strategy. If two or more agents are in \( CL_1 \), the adversary removes edge \( e \); if two or more agents are in \( CL_p \), it removes edge \( e' \). This ensures that whenever at least two agents are in \( CL_1 \), access to \( v_{BH} \) via \( e \) is blocked, and similarly, access via \( e' \) is blocked whenever at least two agents are in \( CL_p \).

Since the agents have $1$-hop visibility and global communication but no knowledge of \( G \) or $n$, they cannot determine that edges \( e \) and \( e' \) lead to node \( v_{BH} \) unless at least one agent traverses those edges. The reason is as follows. Suppose that, according to some algorithm \( \mathcal{A} \), the agents declare that edge \( e \) (resp.\ \( e' \)) from node \( v_1 \) (resp.\ \( v_p \)) leads to \( v_{BH} \) without traversing it. Consider a graph \( G' \) obtained from \( G \) by inserting a node between \( v_{BH} \) and \( v_1 \) (resp.\ \( v_p \)). If the agents were operating in \( G' \) instead of \( G \), their declaration would be incorrect. For instance, in Figure~\ref{fig:imp2}(A), if at some round the agents at node \( v_1 \) in \( G \) declare that edge \( e \) leads to \( v_{BH} \) without traversing it, then in the footprint \( G' \) shown in Figure~\ref{fig:imp2}(B), this declaration becomes incorrect. Hence, such a declaration does not solve the \(1\)-BHSV problem. Therefore, at least one agent must traverse the edge to verify whether \( e \) (resp.\ \( e' \)) leads to \( v_{BH} \).

Agents can access \( v_{BH} \) only in the following two cases:
\begin{itemize}
\item \textbf{Case 1:} Agent \( a_i \), where \( i \in \{1,2,3\} \), is in \( CL_1 \) while \( \{a_1,a_2,a_3\}\setminus\{a_i\} \) are in \( CL_j \) for \( 2 \leq j \leq p \).
\item \textbf{Case 2:} Agent \( a_i \), where \( i \in \{1,2,3\} \), is in \( CL_p \) while \( \{a_1,a_2,a_3\}\setminus\{a_i\} \) are in \( CL_j \) for \( 1 \leq j \leq p-1 \).
\end{itemize}

In Case 1, without loss of generality, suppose that \( a_1 \) moves to \( v_{BH} \) at round \( t \) via edge \( e \) and is destroyed. At some round \( t_1 \geq t \), if \( a_2 \) or \( a_3 \) is at node \( w_1^{(2)} \), the adversary deletes edge \( e_1 \), thereby preventing access to \( CL_1 \). Thus, the remaining agents cannot reach the neighbour of \( v_{BH} \) in \( CL_1 \). The only remaining possibility is to approach \( v_{BH} \) via edge \( e' \). However, if both agents move to \( CL_p \), the adversary deletes \( e' \), and they cannot determine whether edge \( e' \) in \( CL_p \) leads to \( v_{BH} \). Without moving to node \( v_{BH} \), if the agents declare that edge \( e' \) leads to \( v_{BH} \), the declaration may be incorrect. Indeed, consider the graph \( G' \) obtained from \( G \) by inserting a node between \( v_p \) and \( v_{BH} \); in \( G' \), the declaration would be false. Therefore, at least one agent must traverse edge \( e' \) to verify that it leads to \( v_{BH} \). The only possibility to access edge \( e' \) is when exactly one agent is in \( CL_p \). Without loss of generality, assume that \( a_2 \in CL_j \) for \( 2 \leq j \leq p-1 \) and \( a_3 \in CL_p \). If \( a_3 \) moves via edge \( e' \) at round \( t' \) and dies, then at any round \( t \geq t' \), if \( a_2 \) is at \( w_2^{(p-1)} \), the adversary deletes \( e_{p-1} \), and if \( a_2 \) is at \( w_1^{(2)} \), it deletes \( e_1 \). In this way, \( a_2 \) is confined within \( CL_2, \ldots, CL_{p-1} \) and can never reach a neighbour of \( v_{BH} \).

Note that the above argument holds only if there are at least three cliques (i.e., \( p \geq 3 \)). And we may need one node between nodes $v_p$ and $v_{BH}$ (or $v_1$ and $v_{BH}$). Therefore, $n\geq 9$.

The argument for Case 2 is analogous to Case 1. Hence, under this adversary strategy, no agent can reach a neighbour of \( v_{BH} \) at some round $t$, and identify the port leading to it at round $t$. The argument holds regardless of the agents' memory, the nodes' storage, or the availability of 1-hop visibility or global communication. This completes the proof.
\end{proof}

\begin{figure}[h]
    \centering
    \includegraphics[width=0.5\linewidth]{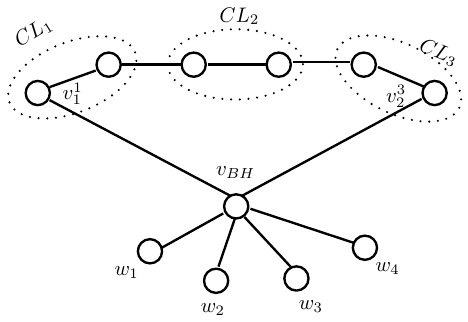}
    \caption{The footprint $G$ for $\delta_{BH}=6$.}
    \label{fig:imp3}
\end{figure}

Now, we show that it is impossible to solve 1-BHSV when $\delta_{BH}+1$ agents are scattered. The idea is to extend the impossibility from the case $\delta_{BH}=2$ to arbitrary $\delta_{BH}\geq 2$. For $\delta_{BH}=2$, the result follows directly from Theorem~\ref{thm:imp1}. For $\delta_{BH}\geq 3$, we attach additional nodes $w_1,\ldots,w_{\delta_{BH}-2}$ to the black hole, thereby increasing its degree while preserving the hard core structure of three cliques used in Theorem~\ref{thm:imp1}. Three agents are confined to this core structure, where, by the previous result, they cannot solve 1-BHSV. The remaining $\delta_{BH}-2$ agents are placed on the nodes $w_i$, each directly connected to the black hole. However, an agent at $w_i$ cannot safely identify whether the incident edge leads to the black hole without traversing it, since otherwise its decision could be incorrect in a modified graph where that edge is subdivided. Thus, each such agent either stays idle or risks destruction by entering the black hole. In either case, these agents do not help in solving the problem. Consequently, even $\delta_{BH}+1$ agents are insufficient to solve 1-BHSV.

\begin{theorem}\label{thm:imp2}
($n \geq 9, \delta_{BH}\geq 2$)
It is impossible for $\delta_{BH}+1$ agents, initially placed at safe nodes of the graph, to solve the \textnormal{1-BHSV} problem, even if the agents have 1-hop visibility, global communication, infinite memory, and the nodes have infinite storage. 
\end{theorem}

\begin{proof}
Let $G$ be a graph on $n$ nodes and let $\delta_{BH}$ be the degree of the black hole node $v_{BH}$.

\begin{itemize}
    \item \textbf{Case 1 $\bm{(\delta_{BH}=2)}$.} 
In the construction of $G$ in Theorem~\ref{thm:imp1}, the value of $\delta_{BH}$ is $2$. Therefore, by Theorem~\ref{thm:imp1}, $\delta_{BH}+1=3$ agents cannot solve \textnormal{1-BHSV}.

\item \textbf{Case 2 $\bm{(\delta_{BH}\geq 3)}$.}  
The construction of $G$ is as follows. Let $w_1, \ldots, w_{\delta_{BH}-2}$ be nodes in $G$. Let $CL_1$, $CL_2$, and $CL_3$ be three cliques of size $2$. Let $CL_i$ contain nodes $v_1^i$ and $v_2^i$ for $i \in [1,3]$. Add edges $(v_1^1, v_{BH})$, $(v_2^1, v_1^2)$, $(v_2^2, v_1^3)$, and $(v_2^3, v_{BH})$. Additionally, add edges $(w_i, v_{BH})$ for all $i \in [1, \delta_{BH}-2]$. It is easy to see that the footprint extends the one we considered in Theorem \ref{thm:imp1}. For $\delta_{BH}=6$, refer to Figure \ref{fig:imp3}. One agent is present at node $w_i$ for every $i\in [1, \delta_{BH}-2]$, and three agents are present at node $v_1^1$. Due to Theorem \ref{thm:imp1}, three agents at node $v_1^1$ cannot solve 1-BHSV. Therefore, only the argument remains for the agent at node $w_i$. If an agent $a_i$ at node $w_i$ without moving through edge $e=(w_i, v_{BH})$ in $G$. Consider a graph \( G' \) obtained from \( G \) by inserting a node between \( v_{BH} \) and \( w_i \). If the agents were operating in \( G' \) instead of \( G \), their declaration would be incorrect. Therefore, agent $a_i$ at node $w_i$ either remains there forever or enters $v_{BH}$ and is destroyed, after which the corresponding node becomes inaccessible to other alive agents. Therefore, $\delta_{BH}+1$ agents cannot solve 1-BHSV in this case as well.
\end{itemize}

Hence, under this adversary strategy, no agent can reach a neighbour of \( v_{BH} \) at round $t$, and identify the port leading to it at round $t$. The argument holds regardless of the agents' memory, the nodes' storage, or the availability of $1$-hop visibility or global communication. The proof is valid for $n \geq 9$, since the minimal case corresponds to $\delta_{BH}=2$, together with Theorem~\ref{thm:imp1}. This completes the proof.
\end{proof}

We now show an impossibility result for 1-interval connected graphs, where the adversary can delete arbitrarily many edges while maintaining $\mathcal{G}_r$ connected. In particular, we prove that $n-2$ agents cannot solve 1-BHSV in this setting. The key idea is that with $n-2$ agents and $n-1$ safe nodes, at least one safe node is always unoccupied. The adversary exploits this by ensuring that, at every round, the only edge connecting the black hole to the rest of the graph is incident to this unoccupied node. Since no agent is present there, no agent can traverse an edge leading to the black hole. Moreover, agents cannot safely declare any edge as leading to the black hole without traversing it, as such a declaration could be incorrect in a slightly modified graph. Thus, the adversary simultaneously blocks exploration and verification. Consequently, no agent ever reaches the black hole or identifies a port leading to it, even though the graph remains connected at every round. Hence, $n-2$ agents are insufficient to solve 1-BHSV.

\begin{theorem}\label{thm:imp3}
It is impossible for $n-2$ agents, initially located at arbitrary safe nodes, to solve the BHSV problem in $1$-interval connected graphs, even if the agents have full visibility, global communication, and infinite memory, and the nodes have infinite storage.
\end{theorem}

\begin{proof}
Let $v_1, v_2, \ldots, v_{n-1}$ be safe nodes and let $v_{BH}$ be the black hole. Let $G$ be the wheel graph $W_n$ with center $v_{BH}$, and let $e_i=(v_i,v_{BH})$ for $1 \le i \le n-1$. Observe that to correctly conclude that an edge $e_i$ leads to the black hole, at least one agent must traverse $e_i$. Otherwise, consider an alternative configuration in which $v_{BH}$ is not a black hole; in such a case, declaring $e_i$ as a black hole edge without traversal would yield incorrect information. Hence, any correct algorithm must ensure that some agent traverses $e_i$.

Now consider any execution with $n-2$ agents. Since there are $n-1$ safe nodes, at every round, at least one safe node contains no agent; assume without loss of generality that this node is $v_1$. At the beginning of each round $r \ge 0$, the adversary constructs $\mathcal{G}_r = (V, E \setminus \bigcup_{i=2}^{n-1} \{e_i\})$. In this graph, the only edge incident to $v_{BH}$ is $e_1$, and since no agent is located at $v_1$, no agent can traverse an edge incident to $v_{BH}$ in any round. Consequently, no agent ever reaches $v_{BH}$ or can determine whether it is a black hole. Also, since $\mathcal{G}_r$ is connected at every round $r$, it maintains 1-interval connectivity. Therefore, the BHSV problem cannot be solved using $n-2$ agents in 1-interval connected graphs, regardless of agent memory, visibility, communication, or node storage. This completes the proof.
\end{proof}

\begin{algorithm}
\caption{Depth-First Search by an agent $a_i$}\label{algo:DFS}
\If{$a_i.state=explore$}
{
    \If{the current node $v$ is already visited by $a_i$}
    {
        set $a_i.prt\_out=a_i.prt\_in$ and move through $a_i.prt\_out$\\
    }
    \Else
    {
        mark the current node $v$ as visited node\\
        set $a_i.prt\_out=(a_i.prt\_in+1)\mod deg(v)$\\
        \If{$a_i.prt\_out=$ value of port used to enter into $v$ for the first time}
        {
            set $a_i.state=backtrack$ and move through $a_i.prt\_out$\\
        }
        \Else
        {
            move through $a_i.prt\_out$\\
        }
    }
}
\ElseIf{$a_i.state=backtrack$}
{
    set $a_i.prt\_out=(a_i.prt\_in+1)\mod deg(v)$\\
        \If{$a_i.prt\_out=$ value of port used to enter into $v$ for the first time}
        {
            set $a_i.state=backtrack$ and move through $a_i.prt\_out$\\
        }
        \Else
        {
            set $a_i.state=explore$ and move through $a_i.prt\_out$\\
        }
}
\end{algorithm}

\section{Algorithm using 1-hop visibility}\label{sec:onehop}
In this section, we present an algorithm for solving the $1$-BHSV problem using four agents that are initially co-located at a single safe node, have $1$-hop visibility and f-2-f communication. Recall that since agents are equipped with 1-hop visibility, they can understand whether the edge corresponding to port $p$ is present.

In~\cite{dyn_disp}, the authors present an exploration algorithm for solving the dispersion problem using two agents in $1$-bounded $1$-interval connected graphs, based on a depth-first search (DFS) traversal by mobile agents. For completeness, we first describe a high-level overview of DFS in this setting. The DFS procedure operates in two states: \textit{explore} and \textit{backtrack}. Each agent $a_i$ maintains the following parameters: $a_i.ID$, which stores its identifier; $a_i.state$, which indicates whether the agent is in the \textit{explore} or \textit{backtrack} state; $a_i.prt\_in$, the port through which the agent enters the current node; and $a_i.prt\_out$, the port through which it exits. Initially, each agent is in the \textit{explore} state, and a node is called the \textit{root} if the agents start their DFS from that node. At each node $v$, the agent updates its state and movement according to Algorithm~\ref{algo:DFS}. It is known that any static graph $G$ with $m$ edges can be explored within $4m$ rounds using DFS~\cite{Das_2018}. Now, based on DFS, we recall the high-level idea of exploration using two agents.

\medskip 
\noindent \textbf{Exploration using two agents:} Consider two agents initially co-located at a node $v$. In \cite{dyn_disp}, authors employ a DFS-based strategy that avoids waiting indefinitely for missing edges by assigning distinct roles to the agents. In particular, one agent, say $a_1$, acts as the leader and strictly follows its original DFS path, while the other agent, say $a_2$, acts as the non-leader. Both agents initiate their DFS from node $v$. At any round, let agent $a_1$ intend to traverse port $p_1$ from node $v_1$ ($a_1.prt\_out=p_1$), and agent $a_2$ intend to traverse port $p_2$ from node $v_2$ ($a_2.prt\_out=p_2$). Note that it is possible that $v_1 = v_2$ and $p_1 = p_2$. The action of each agent depends on $a_i.state$.

\begin{itemize}
    \item If $a_1.state$ is $explore$ (resp. $backtrack$) and the edge corresponding to port $p_1$ is present, then $a_1$ traverses port $p_1$; otherwise, it waits until the edge reappears.
    \item If $a_2.state$ is $explore$ (resp. $backtrack$) and the edge corresponding to port $p_2$ is present, then $a_2$ traverses port $p_2$.
    \item If $a_2.state$ is $explore$ and the edge corresponding to port $p_2$ is missing, then $a_2$ updates its outgoing port as $a_2.prt\_out = (p_2+1)\bmod \deg(v_2)$. If $a_2.prt\_out \neq a_2.prt\_in$, it continues in the $explore$ state and moves through the new port; otherwise, it sets $a_2.state= backtrack$ and moves through the new port.
    \item If the state of $a_2$ is $backtrack$ and the edge corresponding to port $p_2$ is missing, then $a_2$ restarts its DFS from node $v_2$.

    \item If agent $a_2$ is at the root of its current DFS tree and no unexplored ports remain, then it restarts its DFS.
\end{itemize}

We assume that each node is equipped with a whiteboard that stores DFS-related information. Since agent $a_2$ may restart its DFS multiple times, agent $a_2$ and the nodes maintain a counter indicating the number of DFS restarts by agent $a_2$, which allows distinguishing between outdated and current information. This yields a perpetual exploration algorithm, which we denote by \textsc{Exp\_Algo}.

\medskip
We now define the parameters required for the formal description of our algorithm. These parameters are maintained on the whiteboard by the groups and are used to coordinate their actions during the execution. Let $G_1$ and $G_2$ be two groups. Each group contains one leader and one helper.

\begin{itemize}
    \item $\bm{wb_v(G_1).(parent)}$: This parameter stores the information regarding node $v$ w.r.t. the DFS traversal of the leader of $G_1$. The variable $parent$ stores the port used by the leader of the group $G_1$ to visit node $v$ for the first time. At the root of the DFS, $wb_v(G_1).(parent)=-1$.
    \item $\bm{wb_v(G_2).(parent, dfs\_label)}$: It stores the information regarding node $v$ w.r.t. the DFS traversal of the leader of $G_2$. The variable $parent$ stores the port used by the leader of group $G_2$ to visit node $v$ for the first time. The variable $dfs\_label$ stores the number of DFS being run by $G_2$. At the root of the current DFS, $wb_v(G_2).(parent, dfs\_label)=(-1,1)$.
    \item $\bm{Presence\_BH}$: A Boolean parameter indicating whether a black hole is present in the network. If $Presence\_BH=1$, then no black hole exists in the network; otherwise, its value is $0$. Initially, $Presence\_BH=0$.
\end{itemize}

Note that the group $G_1$ does not need to maintain $dfs\_label$ as it never restarts its DFS. It runs only a single DFS, and the leader of $G_1$ always stays on its original path of its DFS. The parameters maintained by each agent $a_i$ are as follows:
\begin{itemize}
    \item $a_i.ID$: It denotes the ID of agent $a_i$.
    \item $a_i.state$: It denotes the state of $a_i$ in its current DFS. Agent $a_i$ can have either an $explore$ or a $backtrack$ state. It is initially set to $explore$.
    \item $a_i.prt\_in$: It stores the port through which an agent $a_i$ enters into the current node. It is initially set to $-1$.  
    \item $a_i.prt\_out$: It stores the port through which an agent $a_i$ exits from the current node. It is initially set to $-1$.
    \item $a_i.dfs\_label:$ This parameter stores the number of times $G_2$ has restarted the DFS. Initially, its value is $1$.  
    \item $a_i.L_{G_1}$: A Boolean parameter indicating whether agent $a_i$ is the leader of group $G_1$. Specifically, $a_i.L_{G_1}=1$ if $a_i$ is the leader of $G_1$, and $a_i.L_{G_1}=0$ otherwise. Initially, $a_i.L_{G_1}=0$.
    \item $a_i.N$: An integer parameter that stores the number of nodes visited by agent $a_i$ during the DFS. Initially, $a_i.N = 0$.
    \item $a_i.H_{G_1}$: A Boolean parameter indicating whether agent $a_i$ is the helper of group $G_1$. Specifically, $a_i.H_{G_1}=1$ if $a_i$ is the helper of $G_1$, and $a_i.H_{G_1}=0$ otherwise. Initially, $a_i.H_{G_1}=0$.
    \item $a_i.L_{G_2}$: A Boolean parameter indicating whether agent $a_i$ is the leader of group $G_2$. Specifically, $a_i.L_{G_2}=1$ if $a_i$ is the leader of $G_2$, and $a_i.L_{G_2}=0$ otherwise. Initially, $a_i.L_{G_2}=0$.
    \item $a_i.H_{G_2}$: A Boolean parameter indicating whether agent $a_i$ is the helper of group $G_2$. Specifically, $a_i.H_{G_2}=1$ if $a_i$ is the helper of $G_2$, and $a_i.H_{G_2}=0$ otherwise. Initially, $a_i.H_{G_2}=0$.
   \item $a_i.(prt, L_{G_1}, r)$: Let $a_i.L_{G_1} = 1$ and agent $a_i$ computes $a_i.prt\_out$. If the edge corresponding to port $a_i.prt\_out$ is missing, then the agent stores the tuple $(prt,1,r)$, where $prt = a_i.prt\_out$ and $r$ denotes the number of consecutive rounds for which the edge corresponding to port $prt$ has been missing. Otherwise, $a_i.(prt, L_{G_1}, r) = (\bot, 0, -1)$.

\item $a_i.(ID_{L_{G_1}}, prt, r')$: If $a_i.L_{G_2}=1$, and it finds that $G_1$ is stuck due to missing edge, then this parameter stores the tuple $(ID_{L_{G_1}}, prt, r)$, where $ID_{L_{G_1}}$ is the identifier of the leader of $G_1$, $prt$ is the port through which $L_{G_1}$ intends to move, and $r'$ denotes the number of consecutive rounds for which the edge corresponding to port $prt$ has been missing. Otherwise, $a_i.(ID_{L_{G_1}}, prt, r') = (\bot,\bot,-1)$.
\end{itemize}

With a slight abuse of notation, we use $L_{G_1}$ (resp.\ $L_{G_2}$) to denote the leader of group $G_1$ (resp.\ $G_2$). Similarly, $H_{G_1}$ (resp.\ $H_{G_2}$) denotes the helper of group $G_1$ (resp.\ $G_2$). Whenever we say that an agent changes its role, it updates its corresponding parameters accordingly. For instance, if an agent $a_i$ acting as $L_{G_2}$ becomes $H_{G_1}$, then it updates its parameters as follows: $a_i.H_{G_1}=1$ and $a_i.L_{G_1}=a_i.L_{G_2}=a_i.H_{G_2}=0$.

\medskip
We first modify \textsc{Exp\_Algo} to ensure that at least one agent can detect that the exploration is complete.

\medskip
\noindent \textbf{Modified exploration strategy:} Let $a_1$ and $a_2$ be two agents initially co-located at a node. Agent $a_1$ acts as the leader of $G_1$ and sets $a_1.L_{G_1}=1$ and $a_1.N=1$, while agent $a_2$ acts as the leader of $G_2$ and sets $a_2.L_{G_2}=1$ and $a_2.N=1$. The algorithm executed by each agent is described as follows. Whenever agent $a_i$ visits a node in $explore$ state, and there is no information on white board based on its current DFS, it appends the value of $a_i.N$ by 1, and whenever it restarts its DFS, then it sets the value of $a_i.N=1$.  

\medskip
\noindent \underline{Steps of $a_1$ (Leader of $G_1$):} Agent $a_1$ follows \textsc{Exp\_Algo} with a slight modification. Let $p$ be the port computed at round $t$ according to \textsc{Exp\_Algo}. If the edge corresponding to port $p$ is present, then $a_1$ sets $a_1.(prt,L_{G_1},r)=(p,1,0)$. Otherwise, if the edge is missing, then $a_1$ sets $a_1.(prt,L_{G_1},r)=(p,1,1)$, and subsequently increases the value of $r$ by $1$ in every round as long as the edge corresponding to $p$ remains missing. When the edge corresponding to $p$ reappears in some round, $a_1$ resets $a_1.(prt,L_{G_1},r)=(\bot,0,-1)$ and traverses port $p$. If $a_1$ is at the root node, and there is no node to explore, it can conclude that it has explored $G$.

\medskip
\noindent \underline{Steps of $a_2$ (Leader of $G_2$):} Agent $a_2$ also follows \textsc{Exp\_Algo} with suitable modifications. When $a_2$ meets $a_1$ at a node $w$ and observes that the edge corresponding to $a_1.prt\_out$ is missing and $a_2.(ID_{L_{G_1}},prt,r')=(\bot,\,\bot,\,-1)$, it restarts its DFS from $w$, treating $w$ as the root of its current DFS tree, and sets $a_2.(ID_{L_{G_1}},prt,r')=(a_1.ID,\,a_1.prt\_out,\,r)$ and $a_2.N=1$, where $r$ is obtained from $a_1.(prt,L_{G_1},r)$. While continuing the DFS from $w$, the following cases can arise.

\begin{itemize}
    \item If $a_2$ encounters $a_1$ at a node $w' \neq w$, it distinguishes this using the whiteboard information, since $w$ is the root of its current DFS tree. If it observes that the edge corresponding to $a_1.prt\_out$ is missing at $w'$, then it restarts its DFS by treating $w'$ as the root of a new DFS tree, and sets $a_2.(ID_{L_{G_1}},prt,r')=(a_1.ID,\,a_1.prt\_out,\,r)$ and $a_2.N=1$, where $r$ is obtained from $a_1.(prt,L_{G_1},r)$.

    \item If $a_2$ is blocked in the backtrack state at a node $w''$ and agent $a_1$ is not present, then it restarts its DFS by treating $w''$ as the root of a new DFS tree, and sets $a_2.(ID_{L_{G_1}},prt,r')=(\bot,\,\bot,\,-1)$ and $a_2.N=1$.

    \item If $a_2$ meets $a_1$ again at the root node $w$ and observes that $r=r'$, then it infers that $a_1$ remains blocked at the same port. Moreover, if all ports at node $w$, except the one corresponding to $a_1$, have been explored, and $r=r'$, then $a_2$ concludes that the exploration is complete.

    \item If $a_2$ meets $a_1$ again at the root node $w$ and observes that $r \neq r'$ and the edge corresponding to $a_1.prt\_out$ is missing, then it restarts its DFS by treating $w$ as the root of a new DFS tree, and sets $a_2.(ID_{L_{G_1}},prt,r')=(a_1.ID,\,a_1.prt\_out,\,r)$ and $a_2.N=1$, where $r$ is obtained from $a_1.(prt,L_{G_1},r)$.
    \item If $a_2$ does not meet $a_1$ at the root node $w$, then it continues its ongoing DFS but sets $a_2.(ID_{L_{G_1}},prt,r')=(\bot,\,\bot,\,-1)$. Moreover, if all ports at node $w$, except the one corresponding to $a_1$, have been explored, then $a_2$ restarts its DFS by treating $w$ as the root of a new DFS tree, and sets $a_2.(ID_{L_{G_1}},prt,r')=(\bot,\,\bot,\,-1)$ and $a_2.N=1$.
\end{itemize}

We denote this algorithm as \textsc{Modified\_Exp\_Algo}. In the analysis, we show that, using \textsc{Modified\_Exp\_Algo}, if group $G_1$ remains blocked on an edge $e$ for $T$ rounds, then agent $a_2$ explores the graph $G \setminus \{e\}$ and correctly concludes that the exploration is complete. This modification transforms the perpetual exploration algorithm of \cite{dyn_disp} into an exploration algorithm in which at least one of the two agents is aware that exploration has been completed. We now present the high-level idea of how this approach can be extended to solve the $1$-BHSV problem using four co-located agents.

\begin{figure}
    \centering
    \includegraphics[width=0.9\linewidth]{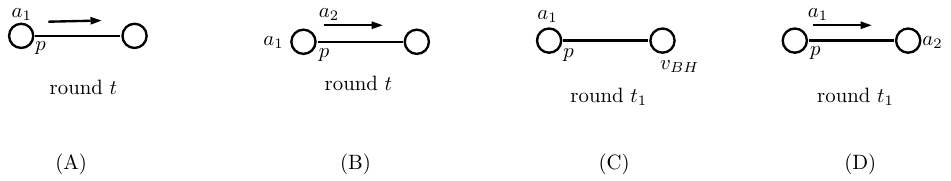}
    \caption{(A) Round $t$ executed by agent $a_1$ according to \textsc{Modified\_Exp\_Algo}. (B)--(D) illustrate how $G_1$ simulates round $t$ of \textsc{Modified\_Exp\_Algo} at round $t$.}
    \label{fig:simulation}
\end{figure}

\medskip
\noindent \textbf{High-level idea to solve 1-BHSV:} Let $a_1, a_2, a_3,$ and $a_4$ be four agents initially located at a safe node of the graph. The agents are partitioned into two groups: $G_1=\{a_1,a_2\}$ and $G_2=\{a_3,a_4\}$. In $G_1$, agent $a_1$ acts as the leader and $a_2$ as the helper, while in $G_2$, agent $a_3$ is the leader and $a_4$ is the helper. We refer to $G_1$ as the leader group and $G_2$ as the non-leader group. Both groups execute \textsc{Modified\_Exp\_Algo}. Since the graph may contain a black hole, agents do not traverse edges blindly and instead employ cautious movement. Suppose that at round $t$, group $G_1$ intends to traverse port $p$ from node $u$ to node $v$ according to \textsc{Modified\_Exp\_Algo}$.$ Refer to Figure~\ref{fig:simulation}(A), where round $t$ is executed by agent $a_1$ when there is no black hole. Now, we show how this movement is simulated by agents $a_1$ and $a_2$ at round $t$.

The helper $a_2$ traverses port $p$ in round $t$, while the leader $a_1$ remains at node $u$ (refer to Figure~\ref{fig:simulation}(B)). At some round $t_1 \geq t+1$, if $a_1$ observes that the edge corresponding to port $p$ is present but $a_2$ is not at node $v$, then $a_1$ concludes that port $p$ leads to the black hole (refer to Figure~\ref{fig:simulation}(C)); otherwise, if $a_1$ observes that the edge is present and $a_2$ is at node $v$, then $a_1$ traverses port $p$ and reaches node $v$ (refer to Figure~\ref{fig:simulation}(D)). Thus, the round $t$ executed by $a_1$ in \textsc{Modified\_Exp\_Algo} is simulated by the group $G_1=\{a_1,a_2\}$. A symmetric cautious movement is performed by the agents of $G_2$. Throughout the execution, both groups update their DFS information on the whiteboard. A careful handling is required in certain scenarios. For instance, suppose $G_1$ is at node $u$ and $G_2$ is at node $v$, and both helpers attempt cautious movement by swapping positions, i.e., $a_2$ moves toward $v$ and $a_4$ moves toward $u$. If in the next round the edge $(u,v)$ is missing, then the leaders cannot reach their helper, and their roles must be reassigned appropriately. 

\medskip
In the next section, we formally explain the technical details of the algorithm.

\subsection{Algorithm}
We divide the algorithm into two phases. The details of each phase is as follows.

\medskip
\noindent\textbf{Phase 1:} Consider four agents $a_1, a_2, a_3,$ and $a_4$ initially co-located at a safe node $v$ at round $r=0$. At round $0$, the roles are initialized as follows: agent $a_1$ sets $a_1.L_{G_1}=1$, agent $a_2$ sets $a_2.H_{G_1}=1$, agent $a_3$ sets $a_3.L_{G_2}=1$, and agent $a_4$ sets $a_4.H_{G_2}=1$. The leaders of both groups initialize the whiteboard at node $v$ by writing $wb_v(G_1).(parent)=-1$ and $wb_v(G_2).(parent,dfs\_label)=(-1,1)$. Both groups then proceed with  \textsc{Modified\_Exp\_Algo}, executing moves cautiously. By cautious movement, we mean that the helper first traverses the port computed by the DFS procedure of its respective group. If this movement is successful, that is, the corresponding edge is present, and the helper is alive at the adjacent node, then the leader traverses the same port. In this way, each group simulates the moves of \textsc{Modified\_Exp\_Algo} using cautious movement. As discussed in the previous section, several cases arise during the execution that require careful handling. In the following, we present all relevant scenarios. Note that whenever agents change their roles, they exchange all the parameters they currently maintain. For instance, if $L_{G_2}$ becomes $H_{G_1}$, then the agent assuming the role of $H_{G_1}$ obtains all the information previously maintained by $H_{G_1}$, and vice versa.

\medskip
\begin{figure}
    \centering
    \includegraphics[width=0.75\linewidth]{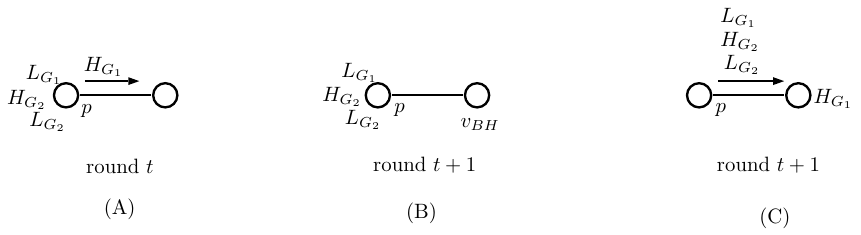}
    \caption{Case (i): $G_1$ and $G_2$ are at the same node at round $t$ and choose the same outgoing port.}
    \label{fig:case1112}
\end{figure}

\begin{figure}
    \centering
    \includegraphics[width=0.45\linewidth]{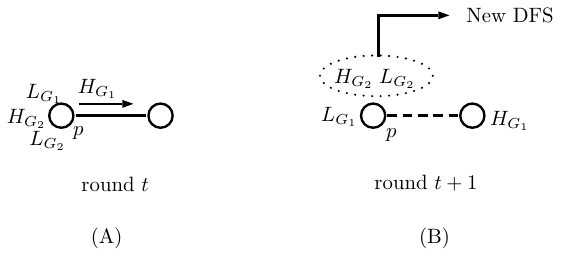}
    \caption{Case (i): $G_1$ and $G_2$ are at the same node at round $t$ and choose the same outgoing port; however, at round $t+1$, the edge corresponding to port $p$ is missing.}
    \label{fig:case111}
\end{figure}

\noindent \textbf{Case (i) (both $\bm{L_{G_1}}$ and $\bm{H_{G_1}}$ are present at a node $\bm{v}$ and have $\bm{state=explore}$)}:
Let the computed $prt\_out$ value by $G_1$ be $p$ at round $t$. Now there are two cases: (i) the edge corresponding to port $p$ is present, or (ii) it is not present. If the edge corresponding to port $p$ is present at round $t$, it does the following. 

\begin{itemize}
    \item[A1.] If $G_1$ finds both $L_{G_2}$ and $H_{G_2}$ at node $v$, and the outgoing port is also $p$, then only $G_1$ performs cautious movement, while $G_2$ waits at node $v$. At round $t$, $H_{G_1}$ moves through port $p$ (refer to Figure \ref{fig:case1112}(A)). If, in round $t+1$, the edge corresponding to port $p$ is present but $H_{G_1}$ is not alive (refer to Figure \ref{fig:case1112}(B)), then all agents at node $v$ conclude that port $p$ leads to the black hole. In round $t+1$, if the edge corresponding to port $p$ is present and $H_{G_1}$ is alive (refer to Figure \ref{fig:case1112}(C)), then all agents at node $v$, including those in $G_2$, traverse port $p$. 

    \medskip
    
    At round $t+1$, if the edge corresponding to port $p$ is not present, then the agents in $G_2$ infer that $G_1$ is blocked due to a missing edge and update their movement according to \textsc{Modified\_Exp\_Algo} (refer to Figure \ref{fig:case111}(A)). Consequently, $L_{G_1}$ updates $(prt, L_{G_1}, r)$ and $L_{G_2}$ updates $(ID_{L_{G_1}}, prt, r')$. As per \textsc{Modified\_Exp\_Algo}, $G_2$ starts a new DFS from node $v$, and $H_{G_1}$ waits for the edge to reappear (refer to Figure \ref{fig:case111}(B)).

\begin{figure}
    \centering
    \includegraphics[width=0.85\linewidth]{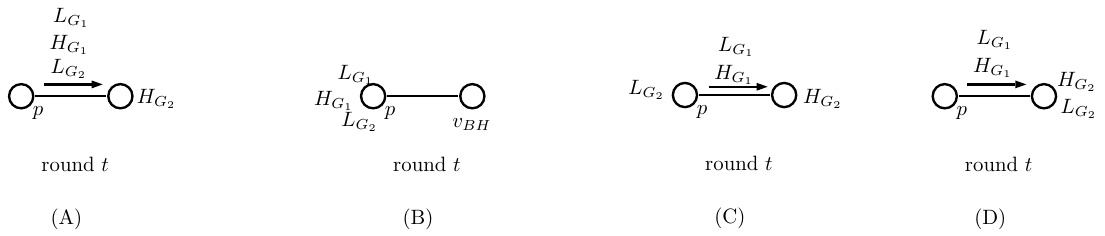}
    \caption{(A) $L_{G_1}$, $H_{G_1}$, $L_{G_2}$ move via port $p$, when they observe $H_{G_2}$ is alive, (B) $L_{G_1}$, $H_{G_1}$, $L_{G_2}$ declare port $p$ leads to $v_{BH}$, when they observe $H_{G_2}$ is not alive, (C) $L_{G_1}$, $H_{G_1}$ move via port $p$ when $L_{G_1}$, $H_{G_1}$, $H_{G_2}$ is co-located, and (D) $L_{G_1}$, $H_{G_1}$ move via port $p$ when $L_{G_1}$, $H_{G_1}$ finds that $L_{G_2}$ and $H_{G_2}$ are present at the node reachable from $v$ via port $p$. }
    \label{fig:case12}
\end{figure}
    \item[A2.] If $G_1$ finds only $L_{G_2}$ at node $v$, and its $prt\_out$ is $p$, then it checks, using $1$-hop visibility, whether $H_{G_2}$ is present at the node corresponding to port $p$. If so, $G_1$ does not perform cautious movement, and $L_{G_1}$ and $H_{G_1}$ move through port $p$ (refer to Figure \ref{fig:case12}(A)). The reason for not performing a cautious movement is that $G_2$ has already executed the cautious movement corresponding to port $p$. Otherwise (refer to Figure \ref{fig:case12}(B)), it concludes that port $p$ leads to the black hole.

    \item[A3.]  If only $H_{G_2}$, is present at the current node and its $prt\_in$ is the same as the $p$ of $G_1$. In this case, $L_{G_1}$ and $H_{G_1}$ move through port $p$ (refer to Figure \ref{fig:case12}(C)). In this case, no additional verification is required, since the cautious movement is initiated by the helper of $G_2$ from a safe node. Therefore, port $p$ leads to a safe node as $H_{G_2}$ is with $G_1$.

    \item[A4.] If $L_{G_2}$ and $H_{G_2}$ are present at the node reachable from $v$ via port $p$ (refer to Figure \ref{fig:case12}(D)), then $L_{G_1}$ and $L_{G_2}$ directly move through port $p$ without performing any cautious movement. 
\end{itemize}

If the edge corresponding to port $p$ is not present at round $t$, it does the following.

\begin{figure}[ht]
    \centering
    \includegraphics[width=0.7\linewidth]{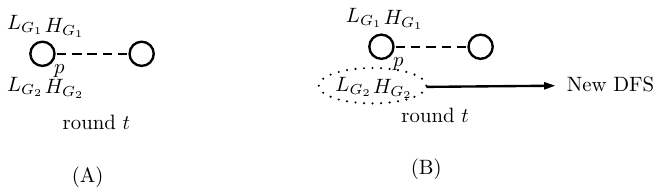}
    \caption{(A) $L_{G_1}$, $H_{G_1}$, $L_{G_2}$, and $H_{G_2}$ are co-located, and the edge corresponding to port $p$ is missing. (B) $G_2$ restarts its DFS, while $G_1$ waits for the edge corresponding to port $p$ to reappear.}
    \label{fig:case13}
\end{figure}

\begin{figure}[ht]
    \centering
    \includegraphics[width=0.7\linewidth]{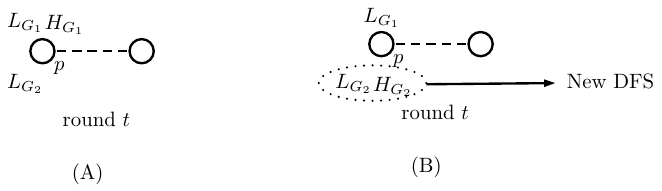}
    \caption{(A) $L_{G_1}$, $H_{G_1}$, and $L_{G_2}$ are co-located, and the edge corresponding to port $p$ is missing. (B) $H_{G_1}$ becomes $L_{G_2}$, $G_2$ restarts its DFS, and $L_{G_1}$ waits for the edge corresponding to port $p$ to reappear.}
    \label{fig:case14}
\end{figure}

\begin{figure}[ht]
    \centering
    \includegraphics[width=0.7\linewidth]{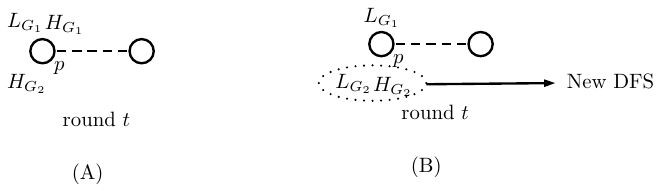}
    \caption{(A) $L_{G_1}$, $H_{G_1}$, and $H_{G_2}$ are co-located, and the edge corresponding to port $p$ is missing. (B) $H_{G_1}$ becomes $L_{G_2}$, $G_2$ restarts its DFS, and $L_{G_1}$ waits for the edge corresponding to port $p$ to reappear.}
    \label{fig:case15}
\end{figure}
\begin{itemize}
    \item[B1.] If $L_{G_1}$ and $H_{G_1}$ finds both $L_{G_2}$ and $H_{G_2}$ at node $v$, and the outgoing port of $G_2$ is also $p$, then only $G_1$ performs cautious movement (refer to Figure \ref{fig:case13} (A)). In this case, the agents in $G_2$ infer that $G_1$ is blocked due to a missing edge and starts a new DFS according as per  \textsc{Modified\_Exp\_Algo} (refer to Figure \ref{fig:case13} (B)). Consequently, $L_{G_1}$ updates $(prt, L_{G_1}, r)$ and $L_{G_2}$ updates $(ID_{L_{G_1}}, prt, r')$, and both groups continue their execution according to \textsc{Modified\_Exp\_Algo}.

\item[B2.] If only $L_{G_2}$ is present at the current node $v$ and its $prt\_out$ is the same as that of $G_1$ (refer to Figure \ref{fig:case14} (A)), then $L_{G_1}$ assigns its helper role to the helper of $G_2$ and infers that the corresponding port has already been explored. In particular, $H_{G_1}$ becomes the helper of $G_2$. In this case, the agents in $G_2$ again infer that $G_1$ is blocked due to a missing edge and starts a new DFS as per \textsc{Modified\_Exp\_Algo} (refer to Figure \ref{fig:case14} (B)). Consequently, $L_{G_1}$ updates $(prt, L_{G_1}, r)$ and $L_{G_2}$ updates $(ID_{L_{G_1}}, prt, r')$, and both groups continue their execution according to \textsc{Modified\_Exp\_Algo}. 

\medskip
\noindent \underline{Note:} In the subsequent round, if the edge corresponding to port $p$ is present and $L_{G_1}$ observes that $H_{G_2}$ is alive at the adjacent node, then $L_{G_1}$ traverses port $p$ and communicates this information. Otherwise, the updated group $G_2$ reaches $H_{G_2}$ and propagates the information accordingly.

    \item[B3.] If only $H_{G_2}$ is present at the current node $v$ and its $prt\_in$ is the same as that of $G_1$ (refer to Figure \ref{fig:case15} (A)), then $H_{G_1}$ becomes $L_{G_2}$. In this case, the agents in $G_2$ again infer that $G_1$ is blocked due to a missing edge and start a new DFS as per \textsc{Modified\_Exp\_Algo} (refer to Figure \ref{fig:case15} (B)). Consequently, $L_{G_1}$ updates $(prt, L_{G_1}, r)$ and $L_{G_2}$ updates $(ID_{L_{G_1}}, prt, r')$, and both groups continue their execution according to \textsc{Modified\_Exp\_Algo}. 

\medskip
\noindent \underline{Note:} In the subsequent round, if the edge corresponding to port $p$ is present and $L_{G_1}$ observes that $L_{G_2}$ is alive at the adjacent node, then $L_{G_2}$ traverses port $p$ and gets this information of group change. Otherwise, the updated group $G_2$ reaches $L_{G_2}$ and propagates the information accordingly.
\item[B4.] In all other cases, $G_1$ continues its wait for the missing edge to reappear.
\end{itemize}

\begin{figure}[ht]
    \centering
    \includegraphics[width=0.5\linewidth]{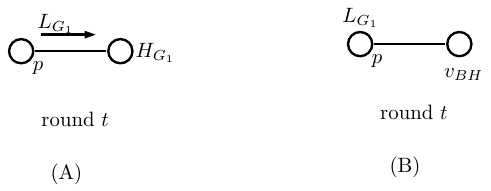}
    \caption{(A) $L_{G_1}$ moves via port $p$, if it finds that $H_{G_1}$ is alive, (B) $L_{G_1}$ declares port $p$ leads to $v_{BH}$, if it does not find at adjacent node corresponding port $p$}
    \label{fig:case16}
\end{figure}

\begin{figure}[ht]
    \centering
    \includegraphics[width=0.75\linewidth]{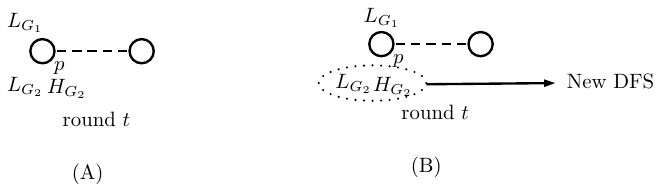}
    \caption{(A) $L_{G_1}$ finds $L_{G_2}$ and $H_{G_2}$ with $prt\_out=p$, (B) $L_{G_2}$, $G_2$ restarts its DFS, and $L_{G_1}$ waits for the edge corresponding to port $p$ to reappear.}
    \label{fig:case17}
\end{figure}

\begin{figure}[ht]
    \centering
    \includegraphics[width=0.5\linewidth]{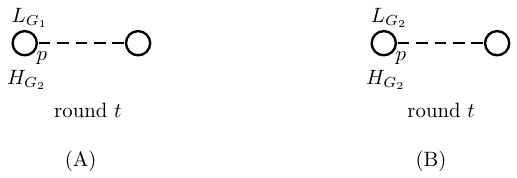}
     \caption{(A) $L_{G_1}$ finds $H_{G_2}$ with $prt\_in=p$, (B) $L_{G_1}$ becomes $L_{G_2}$ declares the completion of the cautious movement of $G_2$.}
    \label{fig:case18}
\end{figure}

\noindent \textbf{Case (ii) (only $\bm{L_{G_1}}$ is present at node $\bm{v}$ and have $\bm{state = explore}$)}: Let $p$ be the port computed by $G_1$ (i.e., $prt\_out = p$) at round $t$. In this case, the cautious movement has already been initiated at round $t$, that is, $H_{G_1}$ has moved through port $p$ in round $t$. At some round $t_1 \geq t+1$, two cases arise: (i) the edge corresponding to port $p$ is present, or (ii) it is not present. 

\medskip
\noindent C*. If the edge corresponding to port $p$ is present at round $t_1$, then $L_{G_1}$ checks, using $1$-hop visibility, whether $H_{G_1}$ is alive at the adjacent node. If $H_{G_1}$ is present, $L_{G_1}$ traverses port $p$ (refer to Figure \ref{fig:case16}(A)); otherwise $L_{G_1}$ concludes that port $p$ leads to the black hole from node $v$(refer to Figure \ref{fig:case16}(B)). In this case, $L_{G_1}$ does not consider whether any agent from $G_2$ is present.

\medskip

If the edge corresponding to port $p$ is not present at round $t_1$, it does the following.

\begin{itemize}
    \item[C1.] If $L_{G_1}$ finds both $L_{G_2}$ and $H_{G_2}$ at node $v$ (refer to Figure \ref{fig:case17}(A)), and the outgoing port is also $p$, then it does the following. In this case, the agents in $G_2$ infer that $G_1$ is blocked due to a missing edge and start a new DFS as per \textsc{Modified\_Exp\_Algo} (refer to Figure \ref{fig:case17} (B)). Consequently, $L_{G_1}$ updates $(prt, L_{G_1}, r)$ and $L_{G_2}$ updates $(ID_{L_{G_1}}, prt, r')$, and both groups continue their execution according to \textsc{Modified\_Exp\_Algo}.

    \item[C2.] If only $L_{G_2}$ is present at the current node $v$, it is definite that the $prt\_out$ values for $L_{G_2}$ and $L_{G_1}$ are different. This is because, as per our algorithm, both $H_{G_1}$ and $H_{G_2}$ can not move through the same port if both groups are together. Suppose it were allowed. If the adjacent node was $v_{BH}$, then both $H_{G_1}$ and $H_{G_2}$ would have died in $v_{BH}$ together. Hence, the leaders of both groups would remain stuck, and neither of them could identify the location of $v_{BH}$. 
    
    \item[C3.] If only $H_{G_2}$ is present at the current node $v$ and its $prt\_in$ is the same as that of $G_1$ (refer to Figure \ref{fig:case18} (A)), then $L_{G_1}$ becomes $L_{G_2}$ (refer to Figure \ref{fig:case18} (B)). This new group, $G_2$, considers the cautious movement to be complete. 

    \item[C4.]  In all other cases, $L_{G_1}$ continues its wait for the missing edge to reappear.
\end{itemize}

\medskip

\begin{figure}
    \centering
    \includegraphics[width=0.15\linewidth]{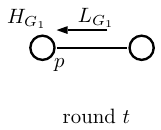}
    \caption{$H_{G_1}$ stays at node $v$, and $L_{G_1}$ moves to node $v$.}
    \label{fig:case19}
\end{figure}

\begin{figure}
    \centering
    \includegraphics[width=0.8\linewidth]{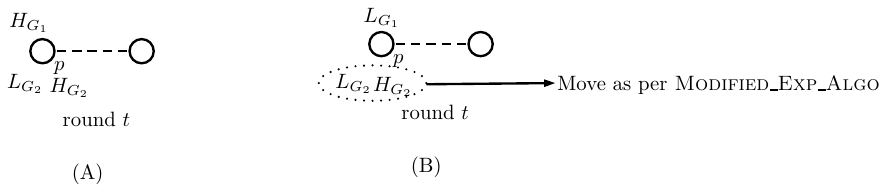}
    \caption{(A) $H_{G_1}$ finds $L_{G_2}$ and $H_{G_2}$ with $prt\_out=p$, (B) $L_{G_2}$ and $H_{G_2}$ continue their movement as per \textsc{Modified\_Exp\_Algo}.}
    \label{fig:case20}
\end{figure}

\begin{figure}
    \centering
    \includegraphics[width=0.5\linewidth]{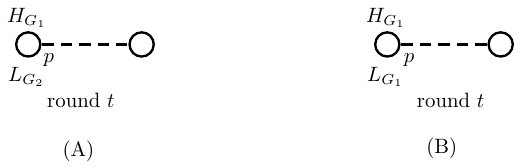}
    \caption{(A) $H_{G_1}$ finds $L_{G_2}$ and $H_{G_2}$ with $prt\_out=p$, (B) $L_{G_2}$ and $H_{G_2}$ continue their movement as per \textsc{Modified\_Exp\_Algo}.}
    \label{fig:case21}
\end{figure}

\noindent \textbf{ Case (iii) (only $\bm{H_{G_1}}$ is present at a node $\bm{v}$ and has $\bm{state=explore}$)}: Let $p$ be the port through which agent $H_{G_1}$ reaches node $v$ (i.e., $prt\_in = p$) at round $t$. In this case, the cautious movement has already been initiated at round $t$. At some round $t_1 \geq t+1$, two cases arise: (i) the edge corresponding to port $p$ is present, or (ii) it is not present.

\medskip
\noindent D*. If the edge corresponding to port $p$ is present at round $t_1$, then $L_{G_1}$ traverses port $p$ (refer to Figure \ref{fig:case19}). In this case, $H_{G_1}$ does not consider whether any agent from $G_2$ is present.

\medskip

If the edge corresponding to port $p$ is not present at round $t_1$, it does the following.

\begin{itemize}
    \item[D1.] If $H_{G_1}$ finds both $L_{G_2}$ and $H_{G_2}$ at node $v$, and the outgoing port is also $p$ (refer to Figure \ref{fig:case20}(A)), then it does the following. In this case, the agents in $G_2$ infer that $G_1$ is blocked due to a missing edge and update their movement according to \textsc{Modified\_Exp\_Algo} (refer to Figure \ref{fig:case20}(B)). In this case, they do not restart their DFS unless they are in the $backtrack$ state.

    \item[D2.] If only $L_{G_2}$ is present at the current node $v$, and its $prt\_out$ is the same as that of $L_{G_1}$ (refer to Figure \ref{fig:case21}(A)), then $H_{G_2}$ becomes $L_{G_1}$ (refer to Figure \ref{fig:case21}(B)). This new group, $G_1$, considers the cautious movement to be complete. 
    
    \item[D3.] If only $H_{G_2}$ is present at the current node $v$, it is definite that the $prt\_in$ values for $H_{G_2}$ and $H_{G_1}$ are different. This is because, as per our algorithm, both $H_{G_1}$ and $H_{G_2}$ can not move through the same port if both groups are together. Suppose it were allowed. If the adjacent node was $v_{BH}$, then both $H_{G_1}$ and $H_{G_2}$ would have died in $v_{BH}$ together. Hence, the leaders of both groups would remain stuck, and neither of them could identify the location of $v_{BH}$. 

    \item[D4.]  In all other cases, $H_{G_1}$ continues its wait for the missing edge to reappear.
\end{itemize}
 
These are all the cases that may occur while performing cautious movement, due to which group exchange may occur. Note that if the state of $G_1$ or $G_2$ is $backtrack$, then they do not need to move cautiously, since the node reached after $backtrack$ has already been explored. Now, let us suppose $G_1$ is at node $u$ and must backtrack via port $p$, which corresponds to edge $(u,v)$. If edge $(u,v)$ is present, $L_{G_1}$ and $H_{G_1}$ traverses port $p$. In this case, $L_{G_1}$ and $H_{G_1}$ do not consider whether any agent from $G_2$ is present.

\medskip

If the edge corresponding to port $p$ is not present at round $t$, as in Case (I), analogous cases arise, as described below.
\begin{itemize}
    \item[E1.] If $G_1$ finds both $L_{G_2}$ and $H_{G_2}$ at node $v$, and the outgoing port is also $p$, then $G_1$ waits for edge to be reappear, while $G_2$ waits at node $v$. In this case, the agents in $G_2$ infer that $G_1$ is blocked due to a missing edge and update their movement according to \textsc{Modified\_Exp\_Algo}. Consequently, $L_{G_1}$ updates $(prt, L_{G_1}, r)$ and $L_{G_2}$ updates $(ID_{L_{G_1}}, prt, r')$, and both groups continue their execution according to \textsc{Modified\_Exp\_Algo}.

\item[E2.] If only $L_{G_2}$ is present at the current node $v$ and its $prt\_out$ is the same as that of $G_1$, then $L_{G_1}$ assigns its helper role to the helper of $G_2$ and infers that the corresponding port has already been explored. In particular, $H_{G_1}$ becomes the helper of $G_2$. In this case, the agents in $G_2$ again infer that $G_1$ is blocked due to a missing edge and update their movement according to \textsc{Modified\_Exp\_Algo}. Consequently, $L_{G_1}$ updates $(prt, L_{G_1}, r)$ and $L_{G_2}$ updates $(ID_{L_{G_1}}, prt, r')$, and both groups continue their execution according to \textsc{Modified\_Exp\_Algo}. 

\medskip
\noindent \underline{Note:} In the subsequent round, if the edge corresponding to port $p$ is present and $L_{G_1}$ finds $H_{G_2}$ to be alive at the adjacent node as $L_{G_1}$ was backtracking. Otherwise, the updated group $G_2$ reaches $H_{G_2}$ and propagates the information accordingly.
    
\item[E3.] If only $H_{G_2}$ is present at the current node $v$ and its $prt\_in$ is the same as that of $G_1$, then $H_{G_1}$ becomes $L_{G_2}$. In this case, the agents in $G_2$ again infer that $G_1$ is blocked due to a missing edge and update their movement according to \textsc{Modified\_Exp\_Algo}. Consequently, $L_{G_1}$ updates $(prt, L_{G_1}, r)$ and $L_{G_2}$ updates $(ID_{L_{G_1}}, prt, r')$, and both groups continue their execution according to \textsc{Modified\_Exp\_Algo}. 

\medskip
\noindent \underline{Note:} In the subsequent round, if the edge corresponding to port $p$ is present and $L_{G_1}$ finds $L_{G_2}$ to be alive at the adjacent node as $L_{G_1}$ was backtracking. Otherwise, the updated group $G_2$ reaches $L_{G_2}$ and propagates the information accordingly.
\item[E4.] In all other cases, $G_1$ continues its wait for the missing edge to reappear.
\end{itemize}

While executing \textsc{Modified\_Exp\_Algo}, if $G_1$ (resp.\ $G_2$) determines that the exploration is complete, then it concludes that there is no black hole in $G$ and proceeds to the next phase. While executing \textsc{Modified\_Exp\_Algo} in Phase 1, if $G_1$ (resp.\ $G_2$) determines at any node $w$ that $Presence\_BH=1$, then it terminates. 

\medskip

\noindent \textbf{Phase 2:} If $G_i$ is in this phase, then the agents execute \textsc{Modified\_Exp\_Algo}. Let $a_1$ and $a_2$ be the agents in $G_i$. Agent $a_1$ assumes the role of $L_{G_1}$ and agent $a_2$ assumes the role of $L_{G_2}$. They execute \textsc{Modified\_Exp\_Algo} without cautious movement. At each visited node, they set the parameter ${Presence\_BH}$ to $1$. After $48 \cdot a_j.N^4$ rounds, for $j \in \{1,2\}$, agent $a_j$ stops the execution of \textsc{Modified\_Exp\_Algo}, and terminates.

In the next section, we provide the correctness of our algorithm.
\subsection{Correctness and analysis of the algorithm}

In this section, we first establish the correctness of \textsc{Modified\_Exp\_Algo}. Let $a_1$ and $a_2$ be two agents initially co-located at a node $v$. Agent $a_1$ acts as the leader of $G_1$ and sets $a_1.L_{G_1}=1$ and $a_1.N=1$, while agent $a_2$ acts as the leader of $G_2$ and sets $a_2.L_{G_2}=1$ and $a_2.N=1$. Suppose that agent $a_1$ is at node $u$ and intends to traverse port $p$ to reach node $v$, but the edge $(u,v)$ is missing at round $t$. We establish the following claim.

\begin{myclaim}\label{claim:G_1}
If the edge $(u,v)$ is missing at every round $t_1 \in [t, t+12m]$, then agent $a_2$ correctly concludes that the exploration is complete. Moreover, $a_2.N = n$.
\end{myclaim}

\begin{proof}
At round $t$, agent $a_2$ may be located at some node $u_1$. Without loss of generality, assume that $u_1 \neq u$. We argue that within the next $8m$ rounds, agent $a_2$ meets agent $a_1$. Between rounds $t$ and $t+4m$, one of the following events occurs: (i) $a_2$ meets $a_1$, (ii) $a_2$ attempts to move from node $v$ to node $u$ while in the backtrack state, (iii) $a_2$ attempts to move from node $v$ to node $u$ while in the explore state, or (iv) $a_2$ reaches the root of its current DFS tree. According to the algorithm, in cases (ii) and (iv), agent $a_2$ restarts its DFS from the corresponding node. In case (iii), agent $a_2$ updates its outgoing port as $(p+1)\bmod \deg(v)$ and continues its DFS according to its current state; within these rounds, this situation eventually transitions to one of the cases (i), (ii), or (iv).

In cases (ii) or (iv), within the subsequent $4m$ rounds, agent $a_2$ reaches node $u$. This follows from the fact that a DFS traversal requires at most $4m$ rounds, and the edge $(u,v)$ remains missing throughout this interval. Moreover, agent $a_2$ cannot starts its new DFS as it never reaches at node $v$ and try to move through edge $(v,u)$ in $backtrack$ state. Let $t_1 \in [t, t+8m]$ be the first round in which $a_2$ meets $a_1$.

Since agent $a_1$ is stuck at node $u$ from round $t$, it sets $a_1.(prt, L_{G_1}, r) = (p,1,1)$ at round $t$, and increments the value of $r$ by $1$ in each subsequent round as long as the edge $(u,v)$ remains missing. Thus, at round $t_1$, we have $a_1.(prt, L_{G_1}, r) = (p,1,t_1 + 1)$. Upon meeting $a_1$, agent $a_2$ sets $(ID_{L_{G_1}}, prt, r') = (a_1.ID, p, t_1 + 1)$ and restarts its DFS from node $u$. Since the edge $(u,v)$ remains missing during the interval $[t_1, t_1+4m] \subseteq [t, t+12m]$, agent $a_2$ is able to explore the graph $G \setminus \{(u,v)\}$ without being blocked at node $v$ and without restarting its DFS. When $a_2$ returns to node $u$ and meets $a_1$ again, it observes that $r = r'$, and hence concludes that the same edge $(u,v)$ has remained missing throughout its traversal. Therefore, it infers that it has explored $G \setminus \{(u,v)\}$.

At round $t_1$, agent $a_2$ initializes $a_2.N = 1$ and increments this value each time it visits a previously unexplored node as per the current DFS. Consequently, by the end of the DFS traversal, it obtains $a_2.N = n$. This completes the proof.
\end{proof}

Based on this claim, we have the following theorem.

\begin{theorem}\label{thm:correctness_exp}
While executing \textsc{Modified\_Exp\_Algo}, at least one of $L_{G_1}$ or $L_{G_2}$ is able to determine that the exploration is complete within first $48m^2$ rounds, and correctly computes the value of $n$.
\end{theorem}
\begin{proof}
Let $P$ be a path which is followed by $L_{G_1}$ as per its DFS. At any node on $P$, suppose the adversary blocks $L_{G_1}$ on an edge $e$ for at most $12m$ consecutive rounds. By Claim~\ref{claim:G_1}, during this period, $L_{G_2}$ is able to explore the graph and correctly conclude that the exploration is complete. To prevent $L_{G_2}$ from completing the exploration, the adversary must allow the edge $e$ to reappear. In that case, $L_{G_1}$ successfully advances by one step along its DFS path. Since a DFS traversal requires at most $4m$ rounds, the length of the DFS path $P$ is at most $4m$. Moreover, at each step along $P$, the adversary can block $L_{G_1}$ for at most $12m$ rounds. Therefore, within $12m \cdot 4m = 48m^2$ rounds, either $L_{G_2}$ concludes that the exploration is complete, or $L_{G_1}$ completes its DFS traversal.

Finally, consider the case where $L_{G_1}$ is blocked at a node $u$ and $L_{G_2}$ restarts its DFS from $u$. It may be possible that $L_{G_1}$ subsequently moves and gets blocked again on the same node by the same port. In this case, $L_{G_2}$ detects the change by observing that the stored values satisfy $r \neq r'$. Furthermore, if $L_{G_2}$ encounters $L_{G_1}$ at a node different from the root of its current DFS, it gets stuck due to the $backtrack$ state, or does not find $L_{G_1}$ at the root, it infers that $L_{G_1}$ has progressed. 

When agent $L_{G_2}$ determines that the exploration is complete, it follows from Claim~\ref{claim:G_1} that its counter satisfies $N = n$. Similarly, agent $L_{G_1}$ increments its counter each time it visits a previously unexplored node; therefore, upon completing the DFS traversal, it also obtains $N = n$. 

Hence, in all cases, at least one of $L_{G_1}$ or $L_{G_2}$ correctly determines that the exploration is complete within $48m^2$ rounds, and correctly computes the value of $n$. This completes the proof.
\end{proof}

 Our algorithm for solving 1-BHSV is based on \textsc{Modified\_Exp\_Algo}, where each move of the exploration procedure is simulated using cautious movement. In this simulation, a helper first attempts to traverse the intended port, and the leader follows only after verifying that the move is safe. However, due to dynamic changes in the graph, such as edge deletions or temporary separations between agents, the execution may deviate from the ideal synchronized behavior of \textsc{Modified\_Exp\_Algo}. To address this, we employ a group exchange mechanism that dynamically reassigns roles among agents. This ensures that, despite disruptions, each group always maintains a valid leader–helper pair and retains all necessary exploration information. Consequently, the combined behavior of the agents remains consistent with a correct execution of \textsc{Modified\_Exp\_Algo} under cautious movement.

\begin{lemma}\label{lm:checking}
    The group exchange mechanism preserves the correctness of the simulation of \textsc{Modified\_Exp\_Algo} under cautious movement.
\end{lemma}
\begin{proof}
    In \textsc{Modified\_Exp\_Algo}, agent $a_1$ acts as $L_{G_1}$ and strictly follows its DFS path, while agent $a_2$ acts as $L_{G_2}$ and may either restart its DFS or skip edges based on the computation. In our 1-BHSV algorithm, we extend this setting by replacing each leader with a group: $G_1$ consists of $L_{G_1}$ and $H_{G_1}$, and $G_2$ consists of $L_{G_2}$ and $H_{G_2}$. Each movement of $L_{G_1}$ (resp.\ $L_{G_2}$) in \textsc{Modified\_Exp\_Algo} is simulated by a cautious movement of $G_1$ (resp.\ $G_2$), which takes two rounds: first, the helper attempts the move, and then the leader follows upon verifying safety. The purpose of this cautious movement is to avoid losing multiple agents simultaneously, since blindly traversing an edge may lead all agents to the black hole $v_{BH}$.

  However, missing edges may interrupt the cautious movement performed by $G_1$ (resp.\ $G_2$). The role of the group exchange mechanism is to guarantee that $G_2$ continues to make progress in every round, even when $G_1$ is blocked on its DFS path. Suppose that at round $t$, $G_1$ is at node $v$, computes $prt\_out = p$, and is in the $explore$ state. In this case, the node, say $v'$, reachable via port $p$ from node $v$ must be visited using cautious movement, as it may not have been previously explored by $G_1$ and could potentially be $v_{BH}$. Let $p'$ be the port at node $v'$ for edge $(v,v')$. The execution may be interrupted in the following ways:

\begin{itemize}
    \item \textbf{Interruption 1 (Missing edge at round $t$):} At round $t$, the edge corresponding to port $p$ is missing, and hence $G_1$ cannot initiate its cautious movement.
    
    \item \textbf{Interruption 2 (Missing edge at round $t+1$):} At round $t$, the edge corresponding to port $p$ is present, and $H_{G_1}$ moves through port $p$. However, at round $t+1$, the edge corresponding to port $p$ becomes missing. In this case, $L_{G_1}$ cannot verify whether $H_{G_1}$ is alive and must wait until the edge reappears. Meanwhile, $H_{G_1}$ (if alive) remains at its current node so that $L_{G_1}$ can eventually reach it.
\end{itemize}

In the case of Interruption 1, it is important to consider the position of $G_2$.

\begin{itemize}
    \item Both $L_{G_2}$ and $H_{G_2}$ are at node $v$, and $prt\_out = p$ (refer to B1). In this case, they observe that $L_{G_1}$ and $H_{G_1}$ are blocked due to the missing edge $(v, v')$. According to our strategy, $G_2$ restarts its DFS from node $v$ and continues execution. This is consistent with \textsc{Modified\_Exp\_Algo}, where whenever $L_{G_1}$ and $L_{G_2}$ are co-located, have the same outgoing port, and the corresponding edge is missing, $L_{G_2}$ restarts its DFS.

\item Only $L_{G_2}$ is at node $v$, and $prt\_out = p$ (refer to B2). In this case, at some round $t' < t$, $G_2$ has already initiated its cautious movement, and $H_{G_2}$ has moved to node $v'$ from node $v$. From round $t'+1$ to $t$, the edge $(v, v')$ remains missing, which leads to this configuration. If no action is taken, both groups may remain blocked. According to our strategy, $H_{G_1}$ takes over the role of $H_{G_2}$, and $G_2$ restarts its DFS from node $v$, treating $G_1$ as blocked due to the missing edge $(v, v')$. This is consistent with \textsc{Modified\_Exp\_Algo}, where whenever $L_{G_1}$ and $L_{G_2}$ are co-located, have the same outgoing port, and the corresponding edge is missing, $L_{G_2}$ restarts its DFS.
    \item Only $H_{G_2}$ is at node $v$, and $prt\_in = p$ (refer to B3). In this case, at some round $t' < t$, $G_2$ has already initiated its cautious movement from node $v'$, and $H_{G_2}$ has moved to node $v$ from node $v'$. From round $t'+1$ to $t$, the edge $(v, v')$ remains missing, which leads to this configuration. If no action is taken, both groups may remain blocked. According to our strategy, $H_{G_1}$ takes over the role of $L_{G_2}$, and $G_2$ restarts its DFS from node $v$, treating $G_1$ as blocked due to the missing edge $(v, v')$. This is consistent with \textsc{Modified\_Exp\_Algo}, where whenever $L_{G_1}$ and $L_{G_2}$ are co-located, have the same outgoing port, and the corresponding edge is missing, $L_{G_2}$ restarts its DFS.
    \item Both $L_{G_2}$ and $H_{G_2}$ are at node $v'$, and $prt\_out = p'$ (refer to B3). In this case, $G_2$ performs a successful movement, since the group is intact. If the edge corresponding to port $p'$ is present, both agents move together; otherwise, they either try the next port or restart their DFS according to the rules of \textsc{Modified\_Exp\_Algo}. Hence, this behavior is consistent with \textsc{Modified\_Exp\_Algo}, where $L_{G_2}$, when alone at a node and encountering a missing edge corresponding to its $prt\_out$, either skips the edge or restarts its DFS.

     \item Only $H_{G_2}$ is at node $v'$, and $prt\_in = p'$. In this case, at some round $t' < t$, $G_2$ has already initiated its cautious movement from node $v$, and $H_{G_2}$ has moved from $v$ to $v'$. Thus, this configuration is identical to Case B2 discussed earlier. Although $H_{G_2}$ is not immediately aware of this change, once the edge reappears, either $L_{G_1}$ can communicate that the role of $H_{G_2}$ has been updated (e.g., it becomes $H_{G_1}$), or the updated group $G_2$ eventually reaches node $v'$ and propagates this information. Hence, consistency of roles and execution is maintained.
    \item Only $L_{G_2}$ is at node $v'$, and $prt\_out = p'$. In this case, at some round $t' < t$, $G_2$ has already initiated its cautious movement from node $v'$, and $H_{G_2}$ has moved from $v'$ to $v$. Thus, this configuration is identical to Case B3 discussed earlier. Although $L_{G_2}$ is not immediately aware of this change, once the edge reappears, either $L_{G_1}$ can communicate that the role of $L_{G_2}$ has been updated (e.g., it becomes $H_{G_1}$), or the updated group $G_2$ eventually reaches node $v'$ and propagates this information. Hence, consistency of roles and execution is maintained.

   \item If $G_2$ is at a node $w \notin \{v, v'\}$, then it successfully executes its movement at round $t$. If $G_2$ is at node $v$ and its $prt\_out \neq p$, then it also executes its movement at round $t$. Similarly, if $G_2$ is at node $v'$ and its $prt\_out \neq p'$, then it successfully executes its movement at round $t$. 

In all these cases, no group exchange is required, and $G_1$ simply waits for the edge $(v, v')$ to reappear. This corresponds to B4.
\end{itemize}

\medskip
Now, we discuss Interruption 2. In this case, it is important to consider the position of $G_2$ at round $t$. 

\begin{itemize}
    \item Both $L_{G_2}$ and $H_{G_2}$ are at node $v$, and $prt\_out = p$ (refer to A1). According to our strategy, $G_2$ does not perform cautious movement together with $G_1$, since in such a case both $H_{G_1}$ and $H_{G_2}$ may traverse the same port and potentially be destroyed at $v_{BH}$, leaving $L_{G_1}$ and $L_{G_2}$ blocked indefinitely. To avoid this, only $G_1$ initiates the cautious movement, while $G_2$ waits at node $v$. 

However, if at round $t+1$ the edge $(v, v')$ is removed by the adversary, then $G_2$ infers that $G_1$ is blocked. According to our strategy, $G_2$ restarts its DFS from node $v$ and continues execution. This is consistent with \textsc{Modified\_Exp\_Algo}, where whenever $L_{G_1}$ and $L_{G_2}$ are co-located, have the same outgoing port, and the corresponding edge is missing, $L_{G_2}$ restarts its DFS.

\item Only $L_{G_2}$ is at node $v$, and $prt\_out = p$ (refer to A2). In this case, at some round $t' < t$, $G_2$ has already initiated its cautious movement, and $H_{G_2}$ has moved from node $v$ to node $v'$. Since the edge $(v, v')$ is present at round $t$, all agents at node $v$ can observe whether $H_{G_2}$ is alive at node $v'$ using $1$-hop visibility. Therefore, either they detect that port $p$ leads to $v_{BH}$, or $G_1$ does not perform cautious movement, as it has already been executed by $G_2$.

    \item Only $H_{G_2}$ is at node $v$, and $prt\_in = p$ (refer to A3). In this case, at some round $t' < t$, $G_2$ has already initiated its cautious movement from node $v'$, and $H_{G_2}$ has moved from $v'$ to $v$. Since the edge $(v, v')$ is present at round $t$, the agents of $G_1$ at node $v$ can safely traverse port $p$ without further verification. This is because node $v'$ is already known to be safe, as it was previously occupied by $G_2$ at round $t'$. Therefore, no additional cautious movement is required in this case.

    \item Both $L_{G_2}$ and $H_{G_2}$ are at node $v'$, and $prt\_out = p'$ (refer to B3). In this case, $G_1$ does not have to perform catious movement, and they can move.

     \item Only $H_{G_2}$ is at node $v'$, and $prt\_in = p'$. In this case, at some round $t' < t$, $G_2$ has already initiated its cautious movement from node $v$, and $H_{G_2}$ has moved from $v$ to $v'$. Thus, this configuration is identical to Case A2 discussed earlier. 
    \item Only $L_{G_2}$ is at node $v'$, and $prt\_out = p'$. In this case, at some round $t' < t$, $G_2$ has already initiated its cautious movement from node $v'$, and $H_{G_2}$ has moved from $v'$ to $v$. Thus, this configuration is identical to Case A3 discussed earlier. 
\end{itemize}

Analogously, a similar set of cases arises when $G_1$ has already initiated its cautious movement at round $t$, and $G_1$ gets stuck at round $t+1$ onward. At some round $t' \geq t+1$, the positions of $G_2$ may be at node $v$ or $v'$. These scenarios are covered by cases C*, C1, C2, C3, C4 and D*, D1, D2, D3, D4. 

In the backtrack state, $G_1$ does not need to perform cautious movement, since node $v'$ has already been visited and is therefore known to be safe. However, $G_1$ may still get blocked due to missing edges, and such situations are handled in cases E1, E2, E3, and E4.

This completes the proof.
\end{proof}

\begin{theorem}\label{thm:1-hop}
The problem of $1$-BHSV can be solved by four agents starting from a rooted initial configuration in $O(n^4)$ rounds when agents are equipped with $1$-hop visibility, $O(\log n)$ memory, and each node has $O(\log n)$ storage.
\end{theorem}

\begin{proof}
Initially, the four agents are partitioned into two groups, $G_1$ and $G_2$, each consisting of two agents. The exploration performed by a single agent is simulated by each group via cautious movement. Consequently, each logical step of the exploration strategy is implemented using two rounds (not necessarily consecutive) of cautious movement.

By Theorem~\ref{thm:correctness_exp}, within $48m^2$ rounds, at least one of $L_{G_1}$ or $L_{G_2}$ visits every node of $G$. Moreover, due to Lemma~\ref{lm:checking}, in every pair of rounds, at least one group successfully performs a movement corresponding to the exploration strategy. Hence, the simulation preserves the progress of \textsc{Modified\_Exp\_Algo}.

If a black hole exists, it is detected within $ 96m^2$ rounds. Indeed, by Theorem~\ref{thm:correctness_exp}, at least one of $L_{G_1}$ or $L_{G_2}$ completes the exploration. Since each step of \textsc{Modified\_Exp\_Algo} is simulated via cautious movement, one of the groups completes the exploration and, during this process, correctly identifies the port leading to the black hole.

If no black hole exists, then by Theorem~\ref{thm:correctness_exp}, one of the groups determines that the exploration is complete and obtains the value of $n$ using their $N$ parameter. Without loss of generality, assume that $G_2$ completes Phase~1, and let $a_1$ and $a_2$ be its agents. In Phase~2, agent $a_1$ assumes the role of $L_{G_1}$ and agent $a_2$ assumes the role of $L_{G_2}$, and they execute \textsc{Modified\_Exp\_Algo} without cautious movement. Again, by Theorem~\ref{thm:correctness_exp}, within $48m^2$ rounds, one of these agents visits every node of $G$. Consequently, at each node, the parameter $Presence\_BH$ is set to $1$. Any agent still in Phase~1 can detect this and terminate. Since $L_{G_1}$ and $L_{G_2}$ know the value of $n$, they can terminate after $48m^2 \leq 48n^4$ rounds, using the bound $m \leq n^2$. Therefore, all agents terminate within $O(n^4)$ rounds.

Finally, we analyze the memory and storage requirements. Each agent maintains a constant number of Boolean parameters, namely $a_i.state$, $a_i.L_{G_1}$, $a_i.H_{G_1}$, $a_i.L_{G_2}$, and $a_i.H_{G_2}$, which require $O(1)$ memory. The remaining parameters, including $a_i.ID$, $a_i.prt\_in$, $a_i.prt\_out$, $a_i.(prt, L_{G_1}, r)$, and $a_i.(ID_{L_{G_1}}, prt, r')$, require $O(\log n)$ memory, since the values $r$ and $r'$ are bounded by $O(n^4)$. Similarly, the DFS label maintained by the agents is at most $O(n^4)$ and can be stored using $O(\log n)$ bits. At each node $v$, the whiteboard stores $wb_v(G_1).(parent)$, $wb_v(G_2).(parent, dfs\_label)$, and $Presence\_BH$. Since $parent$ corresponds to a port number, $Presence\_BH$ is a boolean parameter and $dfs\_label$ is bounded by $O(n^4)$, the total storage required per node is $O(\log n)$. This completes the proof.
\end{proof}

\section{Algorithm using global communication}
In this section, we present an algorithm that solves the $1$-BHSV problem using $\delta_{BH}+2$ agents, initially located at safe nodes of the graph, each equipped with global communication and $0$-hop visibility. Since the agents have $0$-hop visibility, they cannot detect whether an adjacent edge is missing.

Our approach to solving $1$-BHSV consists of two components: an exploration procedure and a cautious movement mechanism. The exploration procedure ensures that every node of the graph is visited by at least one agent, while the cautious movement mechanism guarantees safety by preventing all agents from simultaneously entering the node $v_{BH}$. However, implementing this approach is non-trivial, as the agents are initially positioned arbitrarily in the graph. To address this challenge, we first introduce a cautious movement strategy that can be executed independently by each agent.

\begin{figure}
    \centering
    \includegraphics[width=1\linewidth]{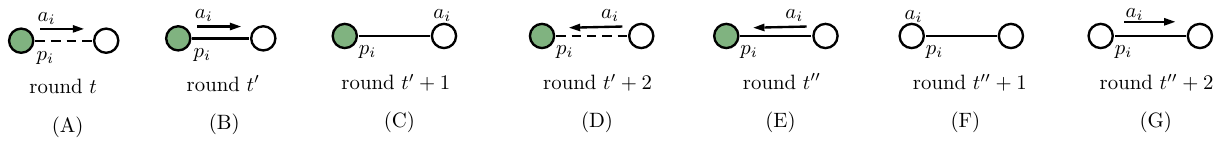}
    \caption{Illustration of a cautious movement performed by agent $a_i$. The green color indicates the node marked by $a_i$, where it records that it is performing cautious movement.}
    \label{fig:CM1}
\end{figure}

\begin{figure}
    \centering
    \includegraphics[width=0.35\linewidth]{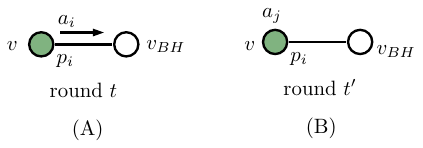}
   \caption{Illustration of a cautious movement performed by agent $a_i$. Upon entering the black hole, the agent is destroyed and leaves permanent information at node $v$.}
    \label{fig:CM12}
\end{figure}
\medskip
\noindent\textbf{Cautious movement ($\bm{CM}$)}: Let an agent $a_i$ be located at a node $v$ at an even round $t$. It computes a port $p_i$ through which it intends to exit $v$. Based on $a_i.state$, its decision is as follows. 

\begin{itemize}
    \item $\bm{a_i.state=explore}:$ It writes this information, along with its unique identifier, on the whiteboard. At the end of round $t$, it attempts to traverse port $p_i$ (refer to Figure
     \ref{fig:CM1}(A)). In round $t+1$, it checks whether the movement was successful. The agent continues attempting to traverse port $p_i$ in every subsequent even round until the movement succeeds. Let $t' \geq t$ be the first even round in which the traversal succeeds (refer to Figure
     \ref{fig:CM1}(B)). 

\vspace{0.15cm}
    \noindent If agent $a_i$ is alive at round $t'+1$ (refer to Figure
     \ref{fig:CM1}(C)), then in round $t'+2$ it attempts to return to node $v$ (refer to Figure
     \ref{fig:CM1}(D)). It continues this return attempt in every subsequent even round until it succeeds. Let $t''$ be the first round in which the agent returns to node $v$ (refer to Figure
     \ref{fig:CM1}(E)). In this case, at round $t''+1$ (refer to Figure
     \ref{fig:CM1}(F)), the agent verifies the outcome, deletes the corresponding information from the whiteboard at node $v$, and proceeds to traverse the port $p_i$ (refer to Figure
     \ref{fig:CM1}(G)).

 \vspace{0.15cm}
    \noindent If agent $a_i$ enters the node $v_{BH}$ and is destroyed  (refer to Figure
     \ref{fig:CM12}(A)), then any subsequent agent $a_j$ arriving at node $v$ can read the whiteboard and deduce that port $p_i$ leads to $v_{BH}$ (refer to Figure
     \ref{fig:CM12}(B)). This follows from the fact that, via global communication, $a_j$ can verify that no agent with identifier $a_i.ID$ is present in the system, implying that $a_i$ has been destroyed at $v_{BH}$. This strategy ensures that not all agents enter the node $v_{BH}$.
    \item $\bm{a_i.state=backtrack}:$ Agent $a_i$ traverses port $p_i$ without writing any information on the whiteboard.
\end{itemize}

\medskip
We now present the exploration algorithm that is used as a subroutine in our approach.

 \medskip
\noindent \textbf{Exploration using two agents}: Let $a_1$ and $a_2$ be co-located at some node. In Section~\ref{sec:onehop}, we presented \textsc{Modified\_Exp\_Algo}, which relies on $1$-hop visibility to detect missing edges. In the current setting, where agents have $0$-hop visibility, we adapt this algorithm as follows. Each round consists of two consecutive rounds: in the even round, an agent attempts to move if required; in the next odd round, it determines whether the movement was successful and broadcasts its ID via global communication. Moreover, the parameters $a_i.(prt, L_{G_1}, r)$ and $a_i.(ID_{L_{G_1}}, prt, r')$ are no longer required. This is because, after attempting a move in an even round, agent $a_2$ can obtain the necessary information during the next odd round via global communication, in particular, whether $a_1$ has successfully left the current node. Thus, the same algorithmic framework can be implemented in this setting by replacing each round of \textsc{Modified\_Exp\_Algo} with two consecutive rounds. By Theorem~\ref{thm:correctness_exp}, the following result holds.

\begin{theorem}\label{thm:exp_}
    While executing \textsc{Modified\_Exp\_Algo}, either $a_1$ or $a_2$ is able to determine that the exploration is complete within first $2\times 48m^2$ rounds.
\end{theorem}

In the next section, we provide the technical details of our algorithm to solve 1-BHSV.

\subsection{Algorithm with correctness analysis}
Let $\{a_1, a_2, a_3, \ldots, a_{\delta_{BH}+2}\}$ be the agents initially placed arbitrarily at safe nodes of $G$. Without loss of generality, assume that $a_i.ID < a_{i+1}.ID$ for every $i \in [1, \delta_{BH}+2]$. Since the agents are equipped with global communication, the two agents with the smallest identifiers (i.e., $a_1$ and $a_2$) initiate the execution of the exploration strategy for two agents, where each edge traversal is replaced by a cautious movement ($CM$) with slight modifications. All other agents store the IDs of the two minimum ID agents that are alive. When these two agents execute \textsc{Modified\_Exp\_Algo}, the agent with the smaller identifier acts as the leader and strictly follows its DFS path, while the other agent acts as the non-leader and may restart its DFS whenever required as per \textsc{Modified\_Exp\_Algo}.

During this process, the following cases arise, and the agents proceed accordingly. 
\begin{figure}
    \centering
    \includegraphics[width=0.5\linewidth]{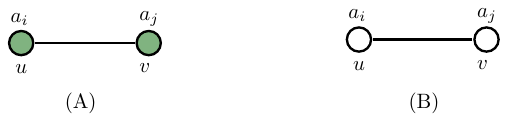}
    \caption{(A) At round $t$, agent $a_i$ writes information at node $v$ and moves to node $u$, while agent $a_j$ writes information at node $u$ and moves to node $v$, and (B) Upon reaching their destinations, agents $a_i$ and $a_j$ delete the information from their respective nodes.}
    \label{fig:placeholder112}
\end{figure}
\begin{itemize}
    \item \underline{Case 1}: Consider an even round $t$. If agent $a_i$ computes the port $p$ at even round $t$ for performing \textsc{Modified\_Exp\_Algo}, it first checks whether any information about port $p$ has been written on the whiteboard by another agent $a_j$ at node $v$. If such information exists, then in the next odd round $t+1$, it checks whether the agent that wrote this information is alive, and also broadcasts its own ID. If it does not hear from that agent, then it deduces that port $p$ leads to the black hole. Otherwise, from round $t+2$ onwards, it starts executing its $CM$. If no information related to $CM$ is available from any other agent, then $a_i$ independently executes its own cautious movement.

    \item \underline{Case 2:} It may happen that, while performing $CM$, at least one of the two agents executing the exploration enters the node $v_{BH}$ and is destroyed. In this case, the remaining agents detect this event via global communication during an odd round. This is because, in every odd round, the agents performing the exploration using $CM$ broadcast their IDs. Whenever the agents observe that at least one of these two agents is no longer alive, the two alive agents with the smallest identifiers begin executing the exploration using $CM$. All agents continuously update the IDs of the two minimum ID agents that are currently alive. Note that while moving, agents remove any DFS information that does not correspond to the two currently alive agents. Moreover, agents maintain a parameter (e.g., $dfs\_label$), as in Section~\ref{sec:onehop}, to distinguish between outdated and current information stored at nodes or in their local memory.

    \item \underline{Case 3:} Suppose agents $a_i$ and $a_j$ are located on opposite endpoints of a missing edge $(u,v)$, with $a_i$ at node $u$ and $a_j$ at node $v$. Agent $a_i$ has written its information at node $v$, and agent $a_j$ has written its information at node $u$ (refer to Figure \ref{fig:placeholder112}(A)). They both can understand this part due to global communication. In this case, $a_i$ deletes the information written by $a_j$ at node $u$, and $a_j$ deletes the information written by $a_i$ at node $v$ (refer to Figure \ref{fig:placeholder112}(B)). Thereafter, both agents continue their DFS traversal. This situation arises when both agents simultaneously initiate their $CM$. Therefore, deleting the information does not compromise the correctness of the execution.

    \item \underline{Case 4:} Consider an even round $t$. If agents $a_i$ and $a_j$ compute the port $p$ at even round $t$ for performing \textsc{Modified\_Exp\_Algo}, then they proceed as follows.

If one of them has $state = backtrack$, then both agents attempt to traverse port $p$ at round $t$. If the movement is successful, then both complete their movement. Otherwise, $a_j$ infers that $a_i$ is stuck and restarts its DFS from round $t+2$, as prescribed by the exploration algorithm.

If both agents have $state = explore$, then they do not traverse port $p$ simultaneously. The agent with the smaller identifier (say $a_i$) first traverses port $p$. If $a_i$ successfully moves through port $p$ in some even round $t_1 \geq t$, then in the next odd round $t_1+1$, agent $a_j$ checks whether it receives any message from $a_i$. If it does not hear from $a_i$ in round $t_1+1$, it deduces that port $p$ leads to $v_{BH}$. Otherwise, in round $t_1+2$, it traverses port $p$. In round $t_1+3$, if $a_j$ observes that its movement is successful, then it completes its cautious movement. Otherwise, it infers that $a_i$ is stuck and restarts its DFS from round $t_1+4$, as prescribed by the exploration algorithm. In this case, the agents do not need to write any information on the whiteboard regarding the port being verified for $CM$.

\item \underline{Case 5:} If $a_1$ (resp.\ $a_2$) determines that the exploration is complete, it broadcasts this information, and all agents, upon receiving it, terminate.
\end{itemize}

Based on this, we have the following theorem.
\begin{theorem}\label{thm:global}
    The problem of 1-BHSV can be solved by $\delta_{BH}+2$ agents starting from an arbitrary initial configuration in $O(\delta_{BH}\cdot |E|^2)$ rounds, when agents are equipped with global communication, $O(\log n)$ memory, and each node has $O(\log n)$ storage. 
\end{theorem}

\begin{proof}
    Consider $\delta_{BH}+2$ agents $a_1, a_2, \ldots, a_{\delta_{BH}+2}$ placed arbitrarily at the start. The two agents with the smallest identifiers, namely $a_1$ and $a_2$, begin executing \textsc{Modified\_Exp\_Algo}, while the remaining agents stay idle. By Theorem~\ref{thm:exp_}, at least one of $a_1$ or $a_2$ visits every node of $G$ within $O(m^2)$ rounds.

If there is no black hole, then at least one of $a_1$ or $a_2$ determines that the exploration is complete, by Theorem~\ref{thm:exp_}. According to our algorithm, the agent that detects completion broadcasts this information, and upon receiving it, all agents terminate.
    
    If there is a black hole, at least one of them must eventually die at $v_{BH}$. The remaining agents detect this through global communication and let the two smallest ID survivors take over the roles of $a_1$ and $a_2$. This process repeats: in each $O(m^2)$ round phase, at most one agent dies at $v_{BH}$ while the other explores the graph and leaves information about the port which leads to $v_{BH}$ while doing $CM$. Because node storage persists and agents can check the validity of stored information through global communication, the system consistently accumulates correct knowledge of which ports lead to $v_{BH}$.  

    At most $\delta_{BH}$ agents die in $v_{BH}$. Once this happens, all neighbours of $v_{BH}$ have correct information pointing to which port leads to $v_{BH}$. One of the two surviving agents then explores the graph, reaches a neighbour of $v_{BH}$, and learns the port information, which leads to $v_{BH}$, thereby solving the problem. Since each phase takes $O(m^2)$ rounds and there can be at most $\delta_{BH}$ deaths, the total running time is $O(\delta_{BH}\cdot m^2)$. 
    
    Since two agents execute \textsc{Modified\_Exp\_Algo}, each node maintains DFS information for at most these two agents; any additional information can be discarded. Thus, each node requires $O(\log n)$ storage, and each agent requires $O(\log n)$ memory due to the arguments in Section~\ref{sec:onehop}.

    In this algorithm, an agent writes port information along with its ID while performing $CM$. In each round, this algorithm is performed by at most two agents, ensuring that the node storage remains $O(\log n)$. The memory requirement of each agent also remains $O(\log n)$, since each agent stores only two minimum ID agents that are alive.
    
    This completes the proof. 
\end{proof}

\section{Algorithm using 1-hop visibility and global communication}
In this section, we present an algorithm that solves the BHSV in the 1-interval connectivity problem using $n-1$ agents, initially positioned arbitrarily at the graph's nodes and equipped with global communication and 1-hop visibility. Recall from Section \ref{sec:model}, agents can sense the missing edge(s) with 1-hop visibility. The high-level idea of our algorithm is as follows. 

\medskip

\noindent \textbf{High-Level Idea:} Assume that in every round $r \geq 0$, each agent has complete knowledge of the snapshot graph $\mathcal{G}_r$ and can communicate globally. The main idea is to coordinate agent movements so that agents spread to previously unoccupied nodes while simultaneously identifying black-hole edges. Consider a path $P = v_1 \sim v_2 \sim \cdots \sim v_\ell$ in $\mathcal{G}_r$ such that $v_1$ contains at least two agents, each $v_i$ for $2 \leq i \leq \ell-1$ contains exactly one agent, and $v_\ell$ is empty. If the agent at $v_i$ moves to $v_{i+1}$ for each $i \in [1, \ell-1]$, then one agent reaches $v_\ell$ without creating any hole, allowing the agents to expand their coverage. Moreover, if an agent $a_{\ell-2}$ records the port used by agent $a_{\ell-1}$ at node $v_{\ell-1}$ to move toward $v_\ell$, then in the next round, using global communication, $a_{\ell-2}$ can infer whether $v_\ell$ is a black hole based on the survival of $a_{\ell-1}$. Hence, whenever an agent dies after leaving a node $v$, at least one surviving agent at $v$ can identify the outgoing port leading to the black hole.

Now, suppose that no node contains more than one agent. Then exactly $n-1$ nodes contain one agent each, and there is a unique empty node $v$. The goal is to determine whether $v$ is a black hole. If there exists a length-2 path $v_1 \sim v_2 \sim v_3$ with $v_3$ being a hole, then the agent $a_1$ at $v_1$ moves to $v_2$ and the agent $a_2$ at $v_2$ moves to $v_3$. Agent $a_1$ records the port taken by $a_2$. In the next round, if agent $a_1$ does not hear from agent $a_2$, then $a_1$ at node $v_2$ can conclude that the recorded port from node $v_2$ leads to the black hole. If no such path exists, then $\mathcal{G}_r$ has a star-like structure where the centre is the unique empty node. In this case, the minimum ID agent $a_i$ moves to the centre, and every other agent records, from its node, which port leads to the centre node. In the next round, if other agents do not hear from $a_i$, then every surviving agent can locally, from their position, know which port leads to a black hole; otherwise, they conclude that the centre node is safe.

The main challenge is how to form a map of $\mathcal{G}_r$ at the beginning of round $r$. In \cite{Ajay_dynamicdisp}, the authors have developed a tool which can form a partial map of $\mathcal{G}_r$. In our algorithm, we use this partial map based on \cite{Ajay_dynamicdisp}. Agents have $1$-hop visibility and global communication. Based on the distribution of agents in $\mathcal{G}_r$, we introduce the following notion.

\begin{definition}\label{def:1}
\textbf{(Connected components of agents in $\mathcal{G}_r$)}
The graph $\mathcal{G}_r=(V,E_r)$ can be partitioned into subgraphs 
$\mathcal{G}_r^1, \mathcal{G}_r^2, \ldots, \mathcal{G}_r^{\ell}$, where 
$\mathcal{G}_r^j = (V_j,E_r^j)$ for $1 \leq j \leq \ell$, such that:
\begin{enumerate}
    \item For every $u \in V_j$, at least one agent is present at $u$.
    \item For all $i \neq j$, $V_i \cap V_j = \emptyset$.
    \item There is no edge $e=(u_1,u_2) \in E_r$ with $u_1 \in V_j$ and $u_2 \in V_i$ for $i \neq j$.
\end{enumerate}
\end{definition}

We denote by $CCA(\mathcal{G}_r)$ the collection of all such subgraphs. Intuitively, $CCA(\mathcal{G}_r)$ represents the connected components of the subgraph of $\mathcal{G}_r$ induced by the nodes that host at least one agent. We use the following parameters in our algorithm.

\begin{itemize}
    \item $a_i.ID$: This parameter records the ID of agent $a_i$.
    \item $a_i.(p_{BH}, ID)$: This parameter records that the agent with identifier $ID$ has taken port $p_{BH}$. Initially, $a_i.(p_{BH}, ID) = (\bot, \bot)$.
    \item $a_i.r:$ This parameter records the round number. Initially, $a_i.r=0$. 

   \item $a_i.\textit{disp}$: A Boolean parameter indicating whether the agents have reached a configuration in which no node of $\mathcal{G}_r$ hosts more than one agent. Initially, $a_i.\textit{disp} = 0$. 
    
    \item $\alpha(v)$: This parameter denotes the number of agents located at node $v$.
    
    \item $ID(v)$: This parameter denotes the set of ID(s) which are present at node $v$.
\end{itemize}

\subsection{Algorithm}
Consider an agent $a_i$ located at node $v \in \mathcal{G}_r$. In round $r \geq 0$, the agent performs the following steps.

\medskip

\noindent \textbf{If} $ \bm{a_i.r(\bmod \;2)=0:}$ If $a_i.disp=1$, it terminates. Otherwise, agent $a_i$ increments its round counter $a_i.r \leftarrow a_i.r + 1$ and then executes the following three phases in round $r$.

\begin{itemize}
    \item \textbf{Phase 1 (1-hop view collection)}: 
If $a_i$ is not the minimum ID agent at node $v$, then it remains idle during this phase. Otherwise, it performs the following steps. Since agent $a_i$ has $1$-hop visibility, it can determine which incident port edges are present or missing. For each port $p \in \{0,1,\ldots,\deg(v)-1\}$ such that the corresponding edge exists, agent $a_i$ performs the following steps.

\begin{itemize}
    \item Let $u$ be the neighbor of $v$ reachable via port $p$.
    
    \item If there is at least one agent present at node $u$, define
    $C_v^{p} = (ID(v), p, ID(u))$.
    
    \item Otherwise (i.e., if node $u$ is a hole), define
    $C_v^{p} = (ID(v), p, \bot)$,
    where $\bot$ indicates that $u$ is a hole.
\end{itemize}

Let $C_v = \bigcup_{p} C_v^{p}$ denote the $1$-hop view of agent $a_i$ at node $v$, where $p$ ranges over all ports of $v$ for which an edge is present. Agent $a_i$ broadcasts $C_v$.
\item \textbf{Phase 2 (Graph reconstruction)}: Let $v_1, v_2, \ldots, v_{\ell}$ be the nodes that contain at least one agent at round $r$, and note that $v \in \bigcup_{i=1}^{\ell}\{v_i\}$. According to Phase~1, the minimum ID agent at node $v_i$ broadcasts its 1-hop view. Hence, agent $a_i$ obtains the information $C_{v_i}$ for every $i \in [1,\ell]$. Using these views, the agent $a_i$ located at node $v$ constructs the map $G'=(V',E')$ as follows.

\begin{itemize}
    \item The node set is defined as $V' = \left\{ ID_{v_i} \mid i \in [1,{\ell}]\right\}$, where $ID_{v_i}$ is the minimum ID from set $ID(v_i)$. 

    \item The edge set $E'$ is constructed as follows. For every pair of tuples $(ID(u_1), p_1, ID(u_2)$ and $(ID(u_2), p_2, ID(u_1)$, agent $a_i$ adds an undirected edge $(ID_{u_1}, ID_{u_2})$ with port labels $\pi(ID_{u_1}, ID_{u_2}) = p_1$ and $\pi(ID_{u_2}, ID_{u_1}) = p_2$. It stores $ID(u_1)$ at node $ID_{u_1}$, and $ID(u_2)$ at node $ID_{u_2}$. 
    \item For each tuple $(ID(u_1), p_1, \bot)$, agent $a_i$ marks port $p_1$ at node $ID_{u_1}$ leads to a hole.
\end{itemize}
\item \textbf{Phase 3 (Move)}: If $G'$ is disconnected, it considers the connected component $G''$ of $G'$ in which node $v$ (location of agent $a_i$ at round $r$) is present. This phase is divided into two cases based on the number of agents at node $v_i$.

\medskip
\noindent \underline{${\exists \;i\in [1,\ell]}$ such that ${\alpha(v_i)\geq 2}$:} Let $w_1, w_2, \ldots, w_{\lambda_1}$ be the multinode nodes in $G'$. Let $b_j$ be the minimum ID agent at node $w_j$. Without loss of generality, let $b_1.ID = \min\{ b_j.ID \mid j \in [1,\lambda_1] \}$. Consider the connected component (say $H$) of $G'$ where node $w_1$ is present. If $H\neq G''$, agent $a_i$ stays idle in this phase. Otherwise, it does the following.

    Let $w'_1, w'_2, \ldots, w'_{\lambda_2}$ be the nodes in $G''$ that have at least one port leading to a hole. Let $b'_j$ be the minimum ID active agent at node $w'_j$. Without loss of generality, let $b'_1.ID = \min\{ b'_j.ID \mid j \in [1,\lambda_2] \}$. Since agent $a_i$ is aware of $G''$, it considers a shortest path $P$ between $w_1$ and $w'_1$. If there are multiple shortest paths between $w_1$ and $w'_1$, it selects the lexicographically smallest among them. Let $P = \overline{u}_1(=w_1) \sim \overline{u}_2 \sim \ldots \sim \overline{u}_{\lambda}(=w'_1)$ be this lexicographically shortest path in $G''$ such that port $p_j$ leads from node $\overline{u}_j$ to node $\overline{u}_{j+1}$ for $1\leq j< \lambda$.
    
    The behaviour of agent $a_i$ is described by the following cases.

\begin{itemize}
    \item If $\lambda \geq 2$, it does the following.
    \begin{itemize}
    \item If $v = \overline{u}_j$ for some $1 \leq j < \lambda-1$, then agent $a_i$ moves to node $\overline{u}_{j+1}$. 
        \item If $v = \overline{u}_{\lambda-1}$, then agent $a_i$ stores
    $a_i.(p_{BH}, ID) = (p_{\lambda}, b_1'.ID)$,
    where $p_{\lambda}$ denotes the minimum available port from node $\overline{u}_{\lambda}$ leading to a hole.
    The agent then moves to node $\overline{u}_{\lambda}$.
    \item If $v = \overline{u}_\lambda$, then agent $a_i$ moves from node $\overline{u}_\lambda$ to a hole via the minimum available port.
    \item In all other cases, agent $a_i$ remains at node $v$.
    \end{itemize}

    \item If $\lambda =1$, one of the port of $\overline{u}_1(=w_1)$ leads to a hole. In this case, it does the following.
    \begin{itemize}
        \item If $v=\overline{u}_1$ and $a_1.ID=b_1.ID$, it moves via the minimum available port that leads to a hole.
        \item  If $v=\overline{u}_1$ and $a_1.ID\neq b_1.ID$, agent $a_i$ stores
    $a_i.(p_{BH}, ID) = (p_1, b_1.ID)$,
    where $p_1$ denotes the minimum available port from node $\overline{u}_1$ leading to a hole.
    \item In all other cases, agent $a_i$ remains at node $v$.
    \end{itemize}
\end{itemize}

\medskip
\noindent \underline{${\nexists \;i\in [1,\ell]}$ such that ${\alpha(v_i)\geq 2}$:} Agent $a_i$ sets $a_i.disp=1$. There are two possible cases based on the connected components of $G'$.
    
    \begin{itemize}
        \item \textbf{Case 1 (no connected component in $G'$ with at least two nodes)}: In this case, $n-1$ nodes are isolated. Let $w'_1, w'_2, \ldots, w'_{n-1}$ be nodes in $G'$ which has one agent. Let $b_j'$ be the agent at node $w_j'$. Without loss of generality, let $b'_1.ID = \min\{ b'_j.ID \mid j \in [1,n-1] \}$. If agent $a_i.ID=b'_1.ID$ moves through the minimum available port, which leads to a hole. Otherwise, it stores $a_i.(p_{BH}, ID) = (p_{\lambda}, b'_1.ID)$, where $p_{\lambda}$ denotes the port from node $v$ which leads to a hole.
        \item \textbf{Case 2 (At least one connected component in $G'$ with at least two nodes)}: If the size of $G''$ is 1, then agent $a_i$ terminates. Otherwise, it proceeds as follows. 
Let $w'_1, w'_2, \ldots, w'_{\lambda_2}$ be the nodes of $G''$ that have at least one port leading to a hole. 
Let $b'_j$ denote the minimum ID agent at node $w'_j$. 
Without loss of generality, assume that $b'_1.\text{ID} = \min \{ b'_j.\text{ID} \mid j \in [1,\lambda_2] \}$. Let $u_1, u_2, \ldots, u_{\lambda_3}$ be the neighbours of node $w'_1$, each containing exactly one agent. 
Let $b''_j$ denote the agent located at node $u_j$. 
Without loss of generality, assume that $b''_1.\text{ID} = \min \{ b''_j.\text{ID} \mid j \in [1,\lambda_3] \}$. The behaviour of agent $a_i$ is determined by the following cases.

     \begin{enumerate}
         \item If $a_i.ID\neq  b_1''.ID$ and $a_i.ID\neq b_1'.ID$, it terminates. 
         \item If $b_1''.ID= a_i.ID$, agent $a_i$ stores $a_i.(p_{BH}, ID) = (p_{\lambda}, b'_1.ID)$, where $p_{\lambda}$ denotes the minimum available port from the node $w'_1$ leading to a hole. Agent $a_i$ moves to node $w'_1$. 
         \item If $b_1'.ID= a_i.ID$, it moves via the minimum available port, which leads to a hole.
     \end{enumerate}
    \end{itemize}
\end{itemize}

\medskip
\noindent \textbf{If} $\bm{a_i.r(\bmod\; 2)\neq 0:}$ Agent $a_i$ increments its round counter $a_i.r \leftarrow a_i.r + 1$ and then executes the following three phases in round $r$.

\begin{itemize}
    \item \textbf{Phase 1 (Sharing 0-hop view)}: If agent $a_i$ is not the minimum ID agent at node $v$, it stays idle in this phase. Otherwise, it broadcasts $ID(v)$.

    \item \textbf{Phase 2 (Collection of 0-hop view)}: Let $v_1, v_2, \ldots, v_{\ell}$ be nodes with at least one agent at round $r$, and note that $v \in \bigcup_{i=1}^{\ell}\{v_i\}$. As per Phase 1, the minimum ID agent at node $v_i$ has broadcast $ID(v_i)$. Let $S=\cup_{i=1}^\ell ID(v_i)$.

    \item \textbf{Phase 3 (Checking Black Hole):} Based on the value of $a_i.disp$ and $a_i.(p_{BH}, ID)$, it does the following. 

    \begin{enumerate}
    \item $a_i.\textit{disp}=0$ and $a_i.(p_{BH},ID)=(\bot,\bot)$:  
    Agent $a_i$ does nothing in this phase.

    \item $a_i.\textit{disp}=0$ and $a_i.(p_{BH},ID)\neq(\bot,\bot)$:  
    If $ID \notin S$, agent $a_i$ declares that port $p_{BH}$ at node $v$ leads to the black hole.  
    Otherwise, it sets $a_i.(p_{BH},ID)=(\bot,\bot)$.

    \item $a_i.\textit{disp}=1$ and $a_i.(p_{BH},ID)=(\bot,\bot)$:  
    Agent $a_i$ terminates.

    \item $a_i.\textit{disp}=1$ and $a_i.(p_{BH},ID)\neq(\bot,\bot)$:  
    If $ID \notin S$, agent $a_i$ declares that port $p_{BH}$ at node $v$ leads to the black hole.  
    Otherwise, it declares that there is no black hole in the system and terminates.
\end{enumerate}
\end{itemize}

\subsection{Correctness and analysis of algorithm}
In this section, we prove the correctness of our algorithm, starting with the correctness of the graph reconstruction.

\begin{lemma}\label{lm:map}
If $r \equiv 0 \pmod{2}$ and all agents are alive, then agent $a_i$ knows $CCA(\mathcal{G}_r)$ at the end of Phase~2 of round $r$. Moreover, for every node in this reconstructed structure, the agent correctly identifies (i) which incident ports lead to holes and (ii) the collection of IDs present at those nodes.
\end{lemma}

\begin{proof}
Let $(w_1,w_2)$ be an edge in $CCA(\mathcal{G}_r)$, and let $\pi(w_1,w_2)=p_1$ and $\pi(w_2,w_1)=p_2$. Since $(w_1,w_2)$ belongs to $CCA(\mathcal{G}_r)$, both $w_1$ and $w_2$ contain at least one agent (by Definition~\ref{def:1}). Let $ID_{w_1}$ and $ID_{w_2}$ denote the minimum identifiers among the agents located at $w_1$ and $w_2$, respectively. Agent $a_i$ receives the $1$-hop views $C_{w_1}$ and $C_{w_2}$, and therefore $ID_{w_1}, ID_{w_2} \in V'$ according to Phase~2. Since $\pi(w_1,w_2)=p_1$ and $\pi(w_2,w_1)=p_2$, the tuples 
$C_{w_1}^{p_1} \in C_{w_1}$ and $C_{w_2}^{p_2} \in C_{w_2}$ are obtained. In Phase~2 these tuples are
$
C_{w_1}^{p_1}=(ID(w_1),p_1,ID(w_2)) \quad \text{and} \quad
C_{w_2}^{p_2}=(ID(w_2),p_2,ID(w_1)).
$ Hence, agent $a_i$ adds an undirected edge $(ID_{w_1},ID_{w_2})$ to $E'$ with port labels $\pi(ID_{w_1},ID_{w_2})=p_1$ and $\pi(ID_{w_2},ID_{w_1})=p_2$. Therefore, every edge of $CCA(\mathcal{G}_r)$ is correctly reconstructed.

Now consider a node $w_1 \in CCA(\mathcal{G}_r)$ such that one of its ports, say $p_1$, leads to a hole $w_2$. Since $w_1 \in CCA(\mathcal{G}_r)$, it contains at least one agent, and thus agent $a_i$ receives the view $C_{w_1}$. Since node $w_2$ is a hole, no view $C_{w_2}$ is received. As port $p_1$ of $w_1$ leads to $w_2$, we have $C_{w_1}^{p_1}=(ID(w_1),p_1,\bot).$ According to the rules of Phase~2, $ID_{w_1}\in V'$, and agent $a_i$ correctly records that port $p_1$ of node $ID_{w_1}$ leads to a hole.

Finally, for each node $ID_{w_1}$, the agent stores the set $ID(w_1)$ corresponding to the identifiers of the agents located at node $w_1$. Hence, agent $a_i$ knows the collection of IDs present at each node of the reconstructed structure. This completes the proof.
\end{proof}

We now show that in every even round in which all agents are alive and a multinode exists, either the number of holes decreases by one or one agent dies at the black hole.

\begin{lemma}\label{lm:corr}
Let $r \ge 0$ be an even round. If $n-1$ agents are alive and there exists at least one multinode, then either one agent dies at the black hole, or the number of nodes without agents decreases by one.
\end{lemma}

\begin{proof}
At an even round $r$, suppose there exists a multinode and that $n-1$ agents are alive. According to the algorithm, by the end of Phase~2 the agents construct the map $G'$. By Lemma~\ref{lm:map}, we have $G' = CCA(\mathcal{G}_r)$.

Let $w_1, w_2, \ldots, w_{\lambda_1}$ be multinode(s) in $G'$. Let $b_j$ denote the minimum ID agent at node $w_j$. Without loss of generality, assume $b_1.ID = \min \{ b_j.ID \mid j \in [1,\lambda_1] \}$. Let $H$ be the connected component of $G'$ that contains node $w_1$. Since some node in $G'$ is a multinode, there must be a node in $H$ with a port leading to a hole. According to the algorithm, all agents agree on a path $P$ in $H$ that leads to such a node. Let $P = u_1(=w_1) \sim u_2 \sim \ldots \sim u_{\lambda}$ be this path such that one of the ports of node $u_{\lambda}$ leads to a hole.

Since $P \subseteq H \subseteq CCA(\mathcal{G}_r)$, at least one agent is present at each node $u_j$ for every $j \in [1,\lambda]$ due to Def. \ref{def:1}. According to the algorithm, the minimum ID agent at node $u_j$ for each $j < \lambda$ moves to node $u_{j+1}$. The minimum ID agent at node $u_{\lambda}$ then moves through the minimum available port that leads to a hole. If this agent does not move to the black hole, then the number of holes decreases by one without creating any new holes. Otherwise, the agent located at $u_{\lambda}$ enters the black hole and is destroyed. This completes the proof.
\end{proof}

Now, we present the final theorem. 

\begin{theorem}\label{thm:global-1-hop}
The BHSV problem in $1$-interval connected graphs can be solved using $n-1$ agents in $O(n)$ rounds, while each agent uses $O(\log n)$ memory.
\end{theorem}
\begin{proof}
By Lemma~\ref{lm:corr}, if $n-1$ agents are alive at the beginning of round $r$ and there exists a multinode, then by the end of round $r$ either one agent dies at the black hole or the number of holes decreases by one. Let $r_1$ be the first even round in which there exists a multinode and an agent dies while following the path $P$ in Phase~3 of our algorithm. According to the algorithm, the agents agree on a path 
$P = u_1 \sim u_2 \sim \ldots \sim u_{\lambda}$ 
such that one of the ports of node $u_{\lambda}$ leads to a hole. Moreover, node $u_1$ is a multinode and each node $u_j$ for $j \in [2,\lambda]$ contains at least one agent.

First, consider the case $\lambda = 1$. In this case, one of the neighbors of node $u_1$ is a hole. Since $u_1$ contains at least two agents, the minimum ID agent, say $b_i$, moves through port $p_{\lambda}$ toward the hole. The remaining agent $a_i$ at node $u_1$ stores $a_i.(p_{BH},ID) = (p_{\lambda}, b_i.ID)$, where $p_{\lambda}$ is the port taken by agent $b_i$. Hence $a_i.(p_{BH},ID) \neq (\bot,\bot)$ and $a_i.\textit{disp} = 0$, since at the beginning of round $r_1$ there exists a multinode. If agent $b_i$ dies at the end of round $r_1$, then in round $r_1+1$ (an odd round), agent $a_i$ does not receive any information from $b_i$. Consequently, according to the algorithm, agent $a_i$ correctly concludes that port $p_{\lambda}$ leads to the black hole.

Now consider the case $\lambda \ge 2$. In this case, the agent at node $u_{\lambda-1}$ records the identifier of the agent at node $u_{\lambda}$ together with the port used by that agent to move toward the hole. In round $r_1+1$, if the agent that moved from $u_{\lambda}$ in round $r_1$ dies at the black hole, then the agent at $u_{\lambda-1}$ does not receive the corresponding identifier information. Therefore, according to the algorithm, it correctly concludes that the corresponding port leads to the black hole. Hence, whenever an agent dies while exploring a hole, the port leading to the black hole is correctly identified.

Therefore, by Lemma~\ref{lm:corr}, within the first $2(n-1)$ rounds either the black hole is detected, or the number of holes decreases until all nodes contain exactly one agent, i.e., the agents reach a dispersed configuration.

Suppose the agents reach such a configuration and let $r \le 2(n-1)$ be the first even round with no multinode. At the beginning of round $r$, all agents set $a_i.\textit{disp}=1$, since they receive no information about a multinode during communication. Two cases arise.

\medskip
\noindent \textbf{Case (i).} No connected component of $CCA(\mathcal{G}_r)$ contains more than one node. In this case, $\mathcal{G}_r$ forms a star graph. Let $v_1,v_2,\ldots,v_{n-1}$ be the nodes each containing exactly one agent at round $r$, $v$ is the center of the star, and each $v_i$ is a pendant node. According to the algorithm, an agent $a_i$ at node $v_i$ moves to node $v$. Another agent $a_j$ at node $v_j(\neq v_i)$ records $a_j.(p_{BH},ID) = (p_j,a_i.ID)$, where $p_j$ is the port from node $v_j$ leading to the hole. In round $r+1$, if agent $a_i$ dies, then no agent receives information from $a_i$ and thus correctly concludes that port $p_{BH}$ leads to the black hole. Otherwise, the agents conclude that no black hole exists and terminate safely.

\medskip
\noindent \textbf{Case (ii).} There exists a connected component $G''$ of $CCA(\mathcal{G}_r)$ containing at least two nodes. This is possible due to Lemma \ref{lm:map}, all agents are able to form a map of $CCA(\mathcal{G}_r)$. All agents located in components of size one terminate. According to the algorithm, a node $w_1'$ in $G''$ is selected such that one of its ports leads to a hole. Let $w$ be a neighbor of $w_1'$ containing an agent $\mathcal{A}$. Agent $\mathcal{A}$ records $\mathcal{A}.(p_{BH},ID) = (p_{\lambda},b'_1.ID)$, where $p_{\lambda}$ is the port taken by agent $b'_1$ from node $w_1'$ toward the hole. In round $r+1$, if $\mathcal{A}$ does not receive information from $b'_1$, it concludes that port $p_{BH}$ leads to the black hole. Otherwise, the agents correctly conclude that no black hole exists and terminate safely.

Therefore, the algorithm correctly solves the BHSV problem in $1$-interval connected graphs. Within the first $2(n-1)$ rounds, either the black hole is detected, or the agents reach a dispersed configuration. From a dispersed configuration, within the next two rounds, the agents either identify the black hole or correctly conclude that no black hole exists. Hence, the algorithm runs in $O(n)$ rounds. It's not difficult: all agents terminate within the next two rounds from round $r$ if no black hole is detected.

Regarding memory, each agent stores the parameters $a_i.ID$, $a_i.(p_{BH},ID)$, $a_i.\textit{disp}$, and the round counter $r$. Since $r = O(n)$, $p_{BH} \le n$, and each ID are polynomial in $n$, each parameter requires $O(\log n)$ bits. Therefore, each agent uses $O(\log n)$ memory.

This completes the proof.
\end{proof}

\section{Conclusion}

In this work, we study the BHSV problem in dynamic graphs under stronger communication and visibility models. We first establish fundamental lower bounds, showing that three co-located agents, $\delta_{BH}+1$ agents, and $n-2$ agents are insufficient to solve 1-BHSV and BHSV under various settings, even when agents are equipped with powerful capabilities such as global communication, 1-hop visibility, and unbounded memory. We then present matching or near-optimal (in terms of number of agents) algorithms under different models. In particular, we design an algorithm that solves 1-BHSV using four agents in the rooted setting with 1-hop visibility, and another algorithm that solves 1-BHSV using $\delta_{BH}+2$ agents in arbitrary configurations with global communication. Furthermore, we provide an algorithm for solving BHSV in 1-interval connected graphs using $n-1$ agents with both global communication and 1-hop visibility. Our results significantly reduce the gap between lower and upper bounds on the number of agents required to solve BHSV in dynamic graphs. They also demonstrate that enhanced capabilities, such as global communication and local visibility, can substantially improve the requirements of agents.

\bibliographystyle{unsrt}
\bibliography{Bib.bib}
\end{document}